\documentclass[a4paper]{article}

\usepackage[LGR, T1]{fontenc}
\usepackage[utf8]{inputenc}
\usepackage[british]{babel}
\usepackage{fullpage}
\usepackage{microtype}
\usepackage[backref=page]{hyperref}

\renewcommand*{\backref}[1]{}
\renewcommand*{\backrefalt}[4]{(%
    \ifcase #1 %
          \or p.~#2%
          \else pp.~#2%
    \fi%
    )}

\usepackage{graphics}
\usepackage{bm}
\usepackage{float}
\usepackage{alphabeta}
\usepackage{diagbox}
\usepackage{bigints}
\usepackage{booktabs}
\usepackage[pdftex]{graphicx}
\usepackage{caption}
\usepackage{amsmath,amsthm,amssymb,amsfonts}
\usepackage{hyperref}
\usepackage{cleveref}
\usepackage{tikz}
\usetikzlibrary{quantikz2}
\usetikzlibrary{shapes,decorations}
\usepackage{tikz-cd}
\usepackage{mathtools}
\usepackage{rotating}
\usepackage{blindtext}
\usepackage{bbold}
\usepackage{url}
\usepackage{doi}
\usepackage{physics}
\usepackage{layout}
\usepackage{setspace}
\usepackage{parskip}
\usepackage{dsfont}
\usepackage{enumerate}
\usepackage{enumitem}
\usepackage{xspace}
\usepackage{pdflscape}
\usepackage{pgfplots}
\pgfplotsset{compat=1.18}
\usepackage{tcolorbox}
\usepackage{braket}
\usepackage{mathdots}
\usepackage{xfrac}
\usepackage{cancel}

\theoremstyle{plain}
\newtheorem{theorem}{Theorem}
\newtheorem*{theorem*}{Theorem}
\newtheorem{lemma}[theorem]{Lemma}
\newtheorem{proposition}[theorem]{Proposition}
\newtheorem{corollary}[theorem]{Corollary}
\newtheorem*{corollary*}{Corollary}

\theoremstyle{definition}
\newtheorem{definition}[theorem]{Definition}

\theoremstyle{remark}
\newtheorem{remark}[theorem]{Remark}

\DeclareMathOperator{\diag}{\mathrm{diag}}

\DeclareMathOperator{\SO}{SO}
\DeclareMathOperator{\Orth}{O}
\newcommand{\Cq}[2]{\mathcal{C}^{(#1)}_{#2}}
\newcommand{\Cn}[1]{\Cq{n}{#1}}
\newcommand{\GMq}[2]{\mathcal{G}^{(#1)}_{#2}}

\newcommand{\Mq}[2]{\mathcal{M}^{(#1)}_{#2}}
\newcommand{\GMn}[1]{\GMq{n}{#1}}

\newcommand{\Mn}[1]{\Mq{n}{#1}}

\newcommand{\U}{\mathcal{U}}
\newcommand{\Even}[1]{\mathcal{E}^{(#1)}}
\newcommand{\En}[1]{\Even{n}}
\newcommand{\Odd}[1]{\mathcal{O}^{(#1)}}
\newcommand{\On}[1]{\Odd{n}}
\newcommand{\Alt}[1]{\mathcal{A}^{(#1)}}
\newcommand{\An}[1]{\Alt{n}}

\newcommand{\fSWAP}{\textsf{fSWAP}}
\newcommand{\SWAP}{\textsf{SWAP}}

\newcommand{\tc}{\tilde{c}}

\newcommand{\setdef}  [2]{\mathopen{ \{ } #1 \mid #2 \mathclose{ \} }}             
\newcommand{\Setdef}  [2]{\left\{ #1 \mid #2 \right\}}             

\usepackage{colonequals}
\newcommand{\defeq}{\colonequals}

\newcommand{\NN}{\mathbb{N}}
\newcommand{\ZZ}{\mathbb{Z}}

\newcommand{\RR}{\mathbb{R}}
\newcommand{\CC}{\mathbb{C}}

\renewcommand{\vec}[1]{\mathbf{#1}}
\newcommand{\veca}{\vec{a}}
\newcommand{\vy}{\vec{y}}
\newcommand{\vz}{\vec{z}}
\newcommand{\vyslashk}{\vec{y}\backslash k}

\newcommand{\ek}{\vec{e}_k}
\newcommand{\vyek}{\vec{y} + \ek}
\newcommand{\Fy}{F_{\vy}}
\newcommand{\Fyek}{F_{\vyek}}
\newcommand{\Fyslashk}{F_{\vyslashk}}

\newcommand{\one}{\mathbb{1}}

\newcommand{\ie}{i.e.~}
\newcommand{\eg}{e.g.~}

\allowdisplaybreaks[4]

\usepackage{authblk}

\title{
\vspace{-1.04cm}
Matchgate hierarchy: \\  \mbox{\Large A Clifford-like hierarchy for
deterministic gate teleportation  in matchgate circuits}
}
\author[1,2]{Angelos Bampounis}
\author[1]{Rui Soares Barbosa}
\author[3]{Nadish de Silva}
\affil[1]{{\normalsize INL -- International Iberian Nanotechnology Laboratory, Braga, Portugal}}
\affil[2]{{\normalsize Centre of Mathematics, Universidade do Minho, Braga, Portugal}}
\affil[3]{{\normalsize Department of Mathematics, Simon Fraser University, Burnaby, B.C., Canada}}

\date{}

\definecolor{amber}{rgb}{1.0, 0.49, 0.0}

\usepackage{soul}

\newcommand{\lmu}{\mathfrak{m}}
\newcommand{\lnu}{\mathfrak{n}}
\newcommand{\llam}{\mathfrak{l}}

\begin{document}

\clearpage

\maketitle

\begin{abstract}\noindent
The Clifford hierarchy, introduced by Gottesman and Chuang in 1999, is an increasing sequence of sets of quantum gates crucial to the gate teleportation model for fault-tolerant quantum computation. Gates in the hierarchy can be deterministically implemented, with increasing complexity, via gate teleportation using (adaptive) Clifford circuits with access to magic states.

We propose an analogous gate teleportation protocol and a related hierarchy in the context of matchgate circuits, another restricted class of quantum circuits that can be efficiently classically simulated but are promoted to quantum universality via access to `matchgate-magic' states.
The protocol deterministically implements any $n$-qubit gate in the hierarchy using adaptive matchgate circuits with magic states, with the level in the hierarchy indicating the required depth of adaptivity and thus number of magic states consumed.
It also provides a whole family of novel deterministic matchgate-magic states.

We completely characterise the gates in the matchgate hierarchy for two qubits,
with the consequence that, in this case, the required number of resource states grows linearly with the target gate's level in the hierarchy.
For an arbitrary number of qubits, we propose a characterisation of the matchgate hierarchy by leveraging the fermionic Stone--von Neumann theorem.
It places a polynomial upper bound on the space requirements for representing gates at each level.
\end{abstract}

\setcounter{secnumdepth}{1}
\setcounter{tocdepth}{2}

\tableofcontents

\clearpage
\section{Introduction}

\subsection{Motivation}
Besides the well-known stabiliser subtheory, where circuits are built out of Clifford gates, another intriguing classically simulable class of quantum circuits consists of those built out of a special set of two-qubit gates, called matchgates, acting on nearest-neighbour qubits.
A matchgate is a two-qubit unitary operator $G(A, B) = A \oplus B$, where unitaries $A, B \in \U(2)$ act on the even and odd parity subspaces, respectively, and have equal determinant.
This class of circuits was originally proposed by Valiant and is closely connected to the problem of counting perfect matchings in graphs~\cite{valiant2001matchgates}.
It turns out to have profound physical significance too, in that these circuits correspond to quantum evolutions of non-interacting (``free'') fermions.
This correspondence is given explicitly by the Jordan--Wigner (JW) transformation, a two-way mapping between
$n$-mode fermionic creation and annihilation operators and $n$-qubit unitary operators.

The analogue of the Gottesman--Knill theorem for matchgate circuits was proved by Valiant~\cite{valiant2001matchgates},
establishing that circuits of nearest-neighbour matchgates can be classically efficiently simulated.
As for Clifford circuits, a way to promote matchgate circuits to quantum universality is using so-called `magic' states \cite{bravyi2006universal,hebenstreit_2019}.
Through a `gate-gadget' construction, it has been shown that any pure fermionic state\footnote{An $n$-qubit state is said to be \emph{fermionic} if it is an eigenstate of the parity operator $Z^{\otimes n}$; equivalently, if it can be prepared by applying a fermionic unitary (see \Cref{sec:matchgate_computation}) to a computational basis state.} that
is not Gaussian, i.e. that cannot be prepared by applying a matchgate circuit to a computational basis state,
can be used as a magic state for matchgate computation~\cite{hebenstreit_2019}.

The promotion of restricted classes of circuits to quantum universality
through quantum gate teleportation was put forward by Shor \cite{shor1997fault} and, in more generality, by Gottesman and Chuang~\cite{Gottesman_1999}.
The latter, in the same work, also introduced the Clifford hierarchy:
an increasing family of sets of unitary gates
that can be implemented \emph{deterministically} using gate teleportation on stabiliser circuits.
They can be executed fault-tolerantly on qubits encoded via the most common stabiliser quantum error-correcting codes.
The level of a unitary in the hierarchy is a measure of the complexity of implementing it in this gate-teleportation model of computation.
The first level of the hierarchy is the Pauli group and the second is the Clifford group.
Mathematically, $(k+1)$-level gates are defined recursively as those unitaries that map Pauli operators to $k$-level gates under conjugation.
Gates in the third level acting on $n$ qubits can be implemented deterministically using magic states of (at most) $2n$ qubits and adaptive Clifford circuits.
More generally, gates at any level in the hierarchy can be implemented fault-tolerantly by a similar construction using gates from strictly lower levels in the hierarchy.

\subsection{Contributions}
We propose an analogous gate teleportation scheme for matchgate circuit computation,
and
develop a corresponding gate hierarchy, which we dub the \emph{matchgate hierarchy} (\Cref{sec:matchgate_hierarchy}) .
The first level consists of normalised real linear combinations of Majorana operators.
These are a set $\{c_\mu\}_{\mu=1}^{2n}$ of $2n$ unitary Hermitian operators
that satisfy the canonical anticommutation relations (\Cref{eq:CAR}).
The second level of the hierarchy consists of the unitaries implemented by generalised matchgate circuits.
Mathematically, these conjugate each Majorana operator to a linear combination of Majoranas,
with the coefficients forming a $2n \times 2n$ orthogonal matrix.\footnotemark 

\footnotetext{Generalised matchgate circuits extend matchgate circuits by adding to the gate set the single-qubit gate $X$,
or equivalently, any Majorana operator; see \Cref{sec:matchgate_computation} for details. 
The matrix of coefficients is special orthogonal for \emph{proper} matchgate circuits; it is orthogonal for generalised matchgate circuits.}

To implement gates in higher levels of the hierarchy, we propose (\Cref{sec:teleportation_protocol}) a gate teleportation protocol inspired by that of Ref.~\cite{Gottesman_1999}.
It is an alternative to the protocol proposed by Hebenstreit et al.~\cite{hebenstreit_2019} to implement the non-matchgate $\SWAP$,
and it applies more broadly to implement a large set of gates \emph{deterministically}.

We explicitly characterise the matchgate hierarchy in the two-qubit case, and draw consequences regarding the efficiency of the protocol (\Cref{sec:2qubit_characterisation}).
For arbitrary number of qubits,
we prove a `hierarchy-aware' fermionic Stone--von Neumann theorem, 
which paves the way for a clearer understanding of the structure of the matchgate hierarchy (\Cref{sec:fermionic_stone_von_neumman_theorem}).

\paragraph{Outline}
The remainder of this article is organised as follows.
\Cref{sec:matchgate_computation} presents the background material and fundamental concepts used throughout the paper:
it begins with an overview of matchgate quantum computation, followed by a brief introduction to the theory of Majorana fermions that highlights their connection to matchgate circuits (via fermionic linear optics).
\Cref{sec:matchgate_hierarchy} introduces the matchgate hierarchy and establishes some of its basic properties.
\Cref{sec:teleportation_protocol} presents the gate teleportation protocol that deterministically implements any $n$-qubit gate in the hierarchy using magic states.
\Cref{sec:2qubit_characterisation} offers a characterisation of the matchgate hierarchy in the case of two-qubit gates.
\Cref{sec:fermionic_stone_von_neumman_theorem} establishes a hierarchy-aware version of a fermionic Stone--von Neumann theorem. 
Finally, \Cref{sec:outlook} concludes with some suggested directions for future research.

\section{Matchgate quantum computation}\label{sec:matchgate_computation}

\subsection{Matchgate circuits}

A \emph{matchgate} is a
two-qubit unitary gate of the form
\begin{equation}\label{eq:G_AB}
    G(A, B) \defeq \begin{pmatrix}
A_{11} & 0 & 0 & A_{12} \\
0 & B_{11} & B_{12} & 0 \\
0 & B_{21} & B_{22} & 0 \\
A_{21} & 0 & 0 & A_{22} 
\end{pmatrix} \in \U(4)
\qquad \text{ for $A, B \in \U(2)$}\quad\text{with
$| A | = | B |$.}
\end{equation}
A \emph{matchgate circuit} is a quantum circuit comprised of matchgates acting only on \emph{nearest-neighbour} qubits, \ie on qubits in consecutive lines.
The $n$-qubit unitaries implemented by matchgate circuits form a subgroup of $\mathcal{U}(2^n)$, which we later denote $\Mn{2}$.
It is the subgroup generated by the unitaries of the form 
\[
G(A,B)_{[k,k+1]} \; \defeq \; \one^{\otimes (k-1)} \otimes G(A,B) \,\otimes\, \one^{\otimes (n-k-1)} ,
\]
with $G(A,B)$ a matchgate.
Note that, unlike the Clifford group, this subgroup is not finite up to phase.

The action of a gate of the form $G(A, B)$ -- thus, in particular, of a matchgate -- amounts to the unitary $A \in \U(2)$ acting on the `even-parity' subspace, spanned by computational basis states $\{\ket{00},\ket{11}\}$, and the unitary $B \in \U(2)$ on the `odd-parity' subspace, spanned by $\{\ket{01}, \ket{10}\}$.
As such, it is a parity-preserving operator.

More generally, one can decompose the Hilbert space of $n$ qubits into a direct sum of the even- and odd-parity subspaces, spanned by computational basis states corresponding to bit strings of even and odd parity, respectively.
These are the two eigenspaces of the parity operator $Z^{\otimes n}$, associated with eigenvalues $+1$ and $-1$, respectively.
An $n$-qubit operator (a $2^n \times 2^n$ complex matrix) $M$ is said to be:
\begin{itemize}
    \item \emph{parity-preserving}, or \emph{even}, if it preserves the parity subspaces, that is,
if $[M,Z^{\otimes n}]=0$, or equivalently $Z^{\otimes n}MZ^{\otimes n} = M$;
    \item \emph{parity-reversing}, or \emph{odd}, if it exchanges the parity subspaces, that is, if $\{M,Z^{\otimes n}\}=0$, or equivalently $Z^{\otimes n}MZ^{\otimes n} = -M$.
\end{itemize}
We write $\Even{n}$ and $\Odd{n}$ for the subspaces of $n$-qubit even and odd operators, respectively.
In the context of matchgate computation, for reasons that will become more apparent in the next section, an operator that is either even or odd is known as a \emph{fermionic} operator.
Both the class of parity-preserving operators and that of fermionic operators are closed under composition and tensor products; that is, $M,M' \in \Even{m}$ and $N \in \Even{n}$ implies $M M' \in \Even{m}$ and $M \otimes N \in \Even{m+n}$ and similarly for fermionic operators.

All matchgate circuits are thus parity-preserving (even), since each matchgate is.
It turns out to be convenient to expand this class of circuits to encompass also their parity-reversing (odd) analogues.
To that end, one may \eg
extend the gate set by including $X_1$, the Pauli $X$ gate acting on the first qubit,\footnote{Instead of $X_1$, one may add (i) any other $X_k$, as in Ref.~\cite{Helsen_2022}, or (ii) any other single Majorana operator $c_\mu$, given in \Cref{eq:JW} ahead, or indeed (iii) any nearest-neighbour `odd matchgate', a two-qubit gate of the form $J(A,B)$, as per \Cref{eq:J_AB} ahead, with $|A|=|B|$ (cf.~\Cref{prop:2qubit_hierarchy_characterisation}).} 
resulting in a class of circuits known as \emph{generalised matchgate circuits}, later denoted by $\GMn{2}$. 
These form a larger subgroup of $\mathcal{U}(2^n)$, which contains only fermionic operators.
A strong motivation for this generalisation is that allowing for the possibility of adjoining ancillary qubits
enlarges the class of unitaries that is implementable by matchgate circuits to that of generalised matchgate circuits.\footnotemark

\footnotetext{\label{fn:ancilla}
 The parity-flipping of the first qubit, $X_1$, enough to generate all odd operations, can be implemented by adding an ancillary mode $0$, applying
the unitary $Z_0Z_1 = Z \otimes Z \otimes \one^{(n-1)}$, and then discarding the unentangled ancilla.
The unitary $Z_0Z_1$ is not only even but it is moreover implementable by a matchgate circuit, since $Z \otimes Z = G(\mathbb{1},-\mathbb{1})$.
 See \cite[Section C]{Spee_2018} for more details using the language of fermions.
 }

(Generalised) matchgate circuits with computational basis state inputs and single-qubit computational basis measurements
were shown to be efficiently simulable by a classical computer \cite{valiant2001matchgates, jozsa2008matchgates, terhal2002classical}.
Surprisingly, this class of circuits gets promoted to quantum universality simply through the ability of swapping two neighbouring qubits, a seemingly trivial resource \cite{jozsa2008matchgates}. 
The added gate, $\SWAP = G(\one,X)$, despite being even, is not a matchgate as it fails the determinant condition.\footnotemark\ 
In fact, adding any other even non-matchgate unitary to the gate set would do just as well \cite{Brod_2011}.
This is analogous to classically simulable Clifford computations being lifted to quantum universality by including any non-Clifford gate, such as the $T$ gate.

\footnotetext{\label{fn:fswap}As it happens, there is a matchgate that almost swaps two qubits, the \emph{fermionic swap} $\fSWAP = G(Z,X)$. 
It doesn't quite swap them as it picks up a $-1$ phase for the component over the computational basis state $\ket{11}$.
In the language of fermions to be introduced below: if both modes being swapped are occupied, swapping them picks up a phase of $-1$, as per the anticommutation relations in \Cref{eq:CARgroup}.}

As in the Clifford setting, a way to promote matchgate circuits to quantum universality is via access to a so-called `magic' state, used in a quantum gate teleportation protocol \cite{bravyi2006universal,hebenstreit_2019}.
Magic states for matchgate computations have been completely characterised in Ref.~\cite[Theorem 2]{hebenstreit_2019}.
But the corresponding protocol for gate teleportation is probabilistic and may require multiple rounds to succeed, consuming multiple copies of the resource magic state; see~\cite[Lemma 4]{hebenstreit_2019}.
Here, we focus on those magic states that can be used to implement gates \emph{deterministically}, in that the implementation always works in a single round \emph{without the need for postselection}.

\subsection{Majorana fermions and fermionic linear optics} \label{sec:majorana_fermions}

An interesting aspect of matchgate circuits is that they correspond to quantum evolutions of non-interacting fermions~\cite{terhal2002classical}. The correspondence is given explicitly through the Jordan--Wigner (JW) transformation~\cite{jordanwigner1928}, which we now summarise.
The ensuing perspective not only justifies the physical relevance of matchgate circuits,
but also comes in tandem with a convenient algebraic language and useful intuitions.

Consider a system of $n$ fermionic modes with each mode $k$ associated with annihilation and creation operators $a_k$ and $a_k^{\dagger}$, respectively.
For a neater, more symmetric presentation,
it is useful to define a set of $2n$ Hermitian unitary operators $\{c_\mu\}_{\mu=1}^{2n}$, known as the \emph{Majorana operators}, given by \[c_{2k-1} \defeq a_k + a^{\dagger}_k  \qquad \text{and} \qquad   c_{2k} \defeq -i(a_k - a^{\dagger}_k).\]
These operators satisfy the canonical anticommutation relations (CAR): for $\mu, \nu \in \{1, \ldots, 2n\}$,
\begin{equation}\label{eq:CAR}
    \{c_{\mu}, c_{\nu}\} = 2 \delta_{\mu\nu} \one \qquad \text{where} \qquad  \{c_{\mu}, c_{\nu}\} \defeq c_{\mu} c_{\nu} + c_{\nu} c_{\mu},
\end{equation}
or equivalently, 
\begin{equation}\label{eq:CARgroup}
c_\mu^2 = \one \qquad \text{and} \qquad c_\mu c_\nu = -c_\nu c_\mu \quad \text{for $\mu \neq \nu$}.
\end{equation}

\Cref{eq:CARgroup} may be regarded as a presentation of a group by generators and relations, where besides the generators $c_\mu$ one implicitly considers an additional generator, $-1$, central and of order $2$.
A general element of this \emph{Majorana group} has the form $\pm c_{\mu_1}\cdots c_{\mu_\lmu}$ with $\mu_1 < \cdots < \mu_\lmu$.
\Cref{eq:CAR} or \Cref{eq:CARgroup} may also be thought of as a presentation of an algebra over the complex field;
its elements are polynomials on the Majorana operators subject to the anticommutation relations, having the generic form
 \[
 \sum_{\lmu=0}^{2n} \, \sum_{\mu_1 < \cdots < \mu_\lmu} \, \alpha_{\mu_1, \dots, \mu_\lmu} c_{\mu_1} \cdots c_{\mu_\lmu}.
\]
 
The Jordan--Wigner transformation \cite{jordanwigner1928,nielsen_2005} translates between fermionic systems and one-dimensional chains of spin-$\sfrac{1}{2}$ particles, also known as qubits, mapping back and forth between $n$-qubit Pauli operators and $n$-mode fermionic operators.
One of its original applications was to solve (\ie calculate the energy spectrum of) certain interacting spin-chain Hamiltonians
by reducing to an equivalent, easy-to-solve non-interacting fermionic Hamiltonian.
But one might also think of the transformation going, in a sense, in the opposite direction:
as a way to simulate fermionic systems in a quantum computer.
The transformation represents $n$-mode fermionic systems in the space of $n$ qubits.
Recall that the Pauli exclusion principle \cite{pauli_1925}, which can be derived from the anticommutation relations, states that no two indistinguishable fermions can occupy the same mode.
Therefore, specifying a state in the Fock basis, or occupancy number basis, requires a single bit of information about each mode, indicating whether it is occupied.
The state of an $n$-mode fermionic system can thus be represented by an $n$-qubit state,
with the computational basis state
\[\ket{\vec{z}} \;=\; \ket{z_1, \ldots, z_n}
\qquad \text{for  $\vec{z} \in \ZZ_2^n$}
\]
representing the Fock state 
\[
\ket{\Psi_{\vec{z}}} \;=\; (a^{\dagger}_1)^{z_1} \cdots (a^{\dagger}_n)^{z_n} \ket{\vec{0}},
\]
where $\ket{\vec{0}}$ denotes the vacuum state.

One may think of the JW transformation as a unitary representation of the Majorana group (or algebra) introduced above.
The Majorana operators $\{c_\mu\}_{\mu=1}^{2n}$ are represented as $n$-qubit unitaries, \ie in $\U(2^n)$,
as follows:
for $k = 1, \ldots, n$, 
\begin{equation}\label{eq:JW}
    c_{2k-1} \;=\; \left(\textstyle\prod\nolimits_{i=1}^{k-1}Z_i\right)X_k, \qquad\qquad c_{2k} \;=\; \left(\textstyle\prod\nolimits_{i=1}^{k-1}Z_i\right)Y_k . 
\end{equation}

The monomials of the form
$c_{\mu_1} \cdots c_{\mu_\lmu}$ with $\mu_1 < \cdots < \mu_\lmu$ are linearly independent
and form a basis of the whole matrix space $M_{2^n}(\CC)$ of $n$-qubit matrices \cite[Theorem 2]{lee1947clifford}.

This language provides an alternative perspective on even and odd operators. An operator is even (resp.\ odd) if it can be written as the linear combination of even (resp.\ odd) degree monomials of Majorana operators.
That is, $\Even{n}$ (resp.\ $\Odd{n}$) is spanned by the monomials of the form above with $\lmu$ even (resp.\ odd).
Even and odd unitaries are significant from a physical perspective, too. 
In various situations of interest, physical states are constrained by a \emph{parity superselection rule},
which forbids coherent superpositions between even and odd particle number parity states \cite{wigner_1952,Vidal_2021};
in other words, physical states are constrained to be an eigenstate of the parity operator $Z^{\otimes n}$, also known as a \emph{fermionic state}.
Parity is then a classically conserved quantity: under any physically realisable evolution the state is confined to remain in the same eigenspace, so that the corresponding unitary operators ought to be parity-preserving (a.k.a.\ even).
However, allowing for the possibility of adjoining ancillary modes enlarges the physically implementable evolutions to also include the parity-reversing unitaries; see \cref{fn:ancilla}.
In any case, the physically allowed operations are parity-defined.
Accordingly, as mentioned in the previous section, unitaries that are either parity-preserving (even) or parity-reversing (odd) are referred to as \emph{fermionic} \cite{Spee_2018, Vidal_2021}.

A particularly simple type of evolution of a fermionic system is that of non-interacting (``free'') fermions \cite{terhal2002classical}, so called because it can be reduced to a single-particle description.
Such evolutions are governed by a quadratic Hamiltonian 
$H = i \sum_{\mu, \nu=1}^{2n} h_{\mu\nu} c_{\mu}c_{\nu}$ 
where $h = (h_{\mu\nu})$ is a $2n \times 2n$ real anti-symmetric matrix.
An $n$-qubit unitary $U$ is called  \emph{Gaussian} (or fermionic linear optical) if it has the form 
$U = e^{i H}$ for such a quadratic Hamiltonian $H$.
It turns out that Gaussian unitaries are precisely those that can be implemented by matchgate circuits \cite{jozsa2008matchgates}.

A key fact -- and in some sense the justification for the term `linear' -- is that any Gaussian unitary $U$ acts linearly on single Majorana operators by conjugation, that is,
\begin{equation}\label{eq: conjugation_by_matchgate}
   U c_\mu U^{\dagger} = \sum_{\nu=1}^{2n} R_{\mu \nu} c_\nu,    
\end{equation}
for some $R \in \SO(2n)$, a real special orthogonal matrix.
Conversely, it can be shown that for any special orthogonal matrix $R \in \SO(2n)$, there exists an anti-symmetric matrix $h$ such that $R = e^h$.
This determines a quadratic Hamiltonian $H$ and therefore a Gaussian unitary $U = e^{i H}$ that applies the rotation $R$ to the Majorana operators via conjugation, \ie that satisfies \Cref{eq: conjugation_by_matchgate}.
In a nutshell, Gaussian operations over $n$ fermionic modes, or equivalently unitaries implemented by $n$-qubit matchgate circuits,
up to a global phase,
are in one-to-one correspondence with $\SO(2n)$, the $2n\times 2n$ special orthogonal matrices.
Note the compactness of this encoding: an $n$-qubit unitary is a $2^n \times 2^n$ complex matrix!
This underlies the efficient simulability of Gaussian operations or matchgate circuits.

The extension to generalised matchgate circuits, by \eg\ adding to the gate set a Majorana operator $c_\mu$,
corresponds at the level of the conjugation in \Cref{eq: conjugation_by_matchgate} to adding reflections,
real orthogonal matrices with determinant $-1$.\footnotemark
This extends the one-to-one mapping to one between unitaries implemented by generalised matchgate circuits and orthogonal matrices in $\Orth(2n)$.
For a more comprehensive explanation of this mapping, see Ref.~\cite[Section 2.2]{Dias_2024}.

\footnotetext{The unitary $c_\lambda$ corresponds the reflection about the $\lambda$-th coordinate axis of $\mathbb{R}^{2n}$, since $c_\lambda c_\mu c_\lambda = (-1)^{\delta_{\lambda \mu}}c_\mu$, yielding the matrix $R = \diag(1, \ldots, 1,-1,1,\ldots,1)$ with the $-1$ on the $\lambda$-th position.}

The terminology extends also to states. 
An $n$-qubit state is called \emph{Gaussian} if it arises as the action of a Gaussian operation on a Fock state,
or equivalently, of a matchgate circuit on a computational basis state.

Finally, alternative equivalent conditions \cite{bravyi2004lagrangian} for
a fermionic unitary or state to be Gaussian are given in terms of the operator
\[
\Lambda_n = \sum_{k=1}^{2n}c_k \otimes c_k.
\]
A fermionic operator $U$ is (generalised) Gaussian if and only if $[\Lambda_n, U^{\otimes 2}] = 0$, while
a fermionic state $\ket{\psi}$ is Gaussian if and only if $\Lambda_n \ket{\psi}^{\otimes 2} = 0$.\footnotemark

\footnotetext{
This suggests a more general definition of Gaussianity for a $2^n \times 2^m$ matrix $M$, which maps the $m$-qubit space to the $n$-qubit space, as $\Lambda_n  M^{\otimes 2} = M^{\otimes 2} \Lambda_m$; cf. Equation (5) in Ref.~\cite{generalized_matchgates}.}

\section{Matchgate hierarchy} \label{sec:matchgate_hierarchy}

\subsection{Definition}
Recall the definition of the \emph{Clifford hierarchy}, introduced by Gottesman and Chuang~\cite{Gottesman_1999} to analyse resources in the gate teleportation model of quantum computation.
For a fixed number $n$ of qubits, it consists in an increasing sequence of sets of $n$-qubit gates $\{\Cn{k}\}_{k\in\NN}$.
The first level of the hierarchy is the $n$-qubit Pauli group and the second level is the Clifford group. The latter is defined as the normalizer of the Pauli group in the unitary group $\U(2^n)$, \ie the unitaries that conjugate Paulis to Paulis.
Subsequent levels of the hierarchy are defined recursively: for $k \geq 1$, $\Cn{k+1}$ is the set of unitaries that map the Paulis into $\Cn{k}$ by conjugation.

The  (generalised) matchgate hierarchy is defined analogously.
\begin{definition}\label{def:matchgate_hierarchy}
    Fix any $n>1$, the number of qubits (or fermionic modes).
    The $n$-qubit \emph{generalised matchgate hierarchy} is the sequence of $n$-qubit gates $\{\GMn{k}\}_{k\in\NN}$ defined as follows:
    \begin{itemize}
    \item the first level is the set of unit-norm, real linear combinations of Majorana operators,
    \[
    \GMn{1} \;\defeq\; \Setdef{\textstyle\sum\nolimits_{\mu=1}^{2n} a_\mu c_\mu}{\veca \in \RR^{2n}, \|\veca\| = 1} ;
    \]
    \item higher levels are defined recursively: for any $k \geq 1$,
    \[
    \GMn{k+1} \;\defeq\; \Setdef{U \in \U(2^n)}{\forall{\mu  \in \{1, \ldots, 2n\}.}\; U c_{\mu} U^\dagger \in \GMn{k}  \cap \Odd{n}}.
    \]
    \end{itemize}
\end{definition}

A couple of comments regarding this definition are in order.
First, in defining the first level,
normalisation of the vector $\veca$ is there to guarantee unitarity: note that real linear combinations of Majoranas are self-adjoint and their square satisfies
\begin{equation} \label{eq:squarelinci}
\left(\sum_\mu a_\mu c_\mu\right)^2
\;=\;
\sum_{\mu}a_\mu^2 \overbrace{c_\mu^2}^{\one} + \sum_{\lambda < \nu} a_\lambda a_\nu \overbrace{(c_\lambda c_\nu + c_\nu c_\lambda)}^{0}
\;=\;
\|a\| \one .
\end{equation}
Hence, we could equivalently define
$\GMn{1} = \Setdef{U \in \U(2^n)}{U = \textstyle\sum\nolimits_{\mu=1}^{2n} a_\mu c_\mu \text{ for }\veca \in \RR^{2n}}$.

Secondly, in the recursive clause, the restriction that the Majorana operators $c_\mu$ be conjugated to odd operators, \ie that $U c_\mu U^\dagger \in \Odd{n}$, might strike one as odd.
The purpose of this requirement, as \Cref{lemma:unitary_parity_preserving_conjugation_implies_evenorodd} below elucidates, is to force all the gates in the hierarchy to be fermionic.
Thus, equivalently, one could have set
\[\GMn{k+1} \;\defeq\; \Setdef{U \in \Even{n} \cup \Odd{n}}{\forall\mu.\; U c_{\mu} U^\dagger \in \GMn{k}} .\]
In earlier versions of this work \cite{matchgatehierarchy_qpl2023},
we introduced the matchgate hierarchy without imposing this restriction.
That also matches the definition adopted by Cudby and Strelchuk~\cite{matchgate_hierarchy_sergii},
of which we became aware while preparing the present version of this text.
Subsequent results required us to add the further caveat of restricting to the \emph{fermionic} unitaries in the hierarchy.

Two reasons motivate such restriction to fermionic operators.
The first is foundational: as mentioned earlier, the valid physical operations are parity constrained and correspond to fermionic operators.
The second reason is more practical and directly related to the gate teleportation protocol that motivates this hierarchy, presented in \Cref{sec:teleportation_protocol} ahead.
The magic states used as a resource to implement non-Gaussian gates must be swapped into the appropriate place in the circuit where the gate is to be applied. Such swapping must be accomplished via the use of `free' gates, \ie matchgate circuits, only -- which crucially do not include the $\SWAP$ gate. 
By \cite[Lemma 1]{hebenstreit_2019}, this is doable if and only if the state is fermionic.
Accordingly, the gates that can be implemented by the gate teleportation protocol using a fermionic resource state must themselves be fermionic; see \Cref{rem:using_matchgate_magic} ahead for more details.

So, the following lemma characterises fermionic gates as precisely those that preserve matrix parity under conjugation, justifying the requirement that the Majorana operators, which are odd, be conjugated to odd operators.

\begin{lemma}\label{lemma:unitary_parity_preserving_conjugation_implies_evenorodd}
Let $U \in \mathcal{U}(2^n)$ be an $n$-qubit unitary.
If conjugation by $U$ preserves the parity of operators,
in that it maps even operators to even operators and odd operators to odd operators, i.e.\
\[M \in \Even{n} \implies U M U^\dagger \in \Even{n} \qquad \text{ and } \qquad M \in \Odd{n} \implies U M U^\dagger \in \Odd{n},\]
then $U$ is fermionic (\ie either even or odd).
\end{lemma}
\begin{proof}
An arbitrary $n$-qubit operator $M$ can be decomposed as a sum of its even and odd parts:
$M = M_E + M_O$
where
\[
    M_E = \frac{M + Z^{\otimes n} M Z^{\otimes n}}{2} \qquad \text{ and } \qquad M_O = \frac{M - Z^{\otimes n} M Z^{\otimes n}}{2}.
\]
Conjugation of such an operator $M$ by $U$ yields
\begin{align*}
    U M U^\dagger &= U\Bigg(\frac{M + Z^{\otimes n} M Z^{\otimes n}}{2} + \frac{M - Z^{\otimes n} M Z^{\otimes n}}{2}\Bigg)U^{\dagger} \\
    &= U\Bigg(\frac{M + Z^{\otimes n} M Z^{\otimes n}}{2}\Bigg)U^{\dagger} + U\Bigg(\frac{M - Z^{\otimes n} M Z^{\otimes n}}{2}\Bigg)U^{\dagger} \\
    &= Z^{\otimes n} U\Bigg(\frac{M + Z^{\otimes n} M Z^{\otimes n}}{2}\Bigg)U^{\dagger} Z^{\otimes n} - Z^{\otimes n}U\Bigg(\frac{M - Z^{\otimes n} M Z^{\otimes n}}{2}\Bigg)U^\dagger{2}Z^{\otimes n} \\
    &= Z^{\otimes n}  U (Z^{\otimes n} M Z^{\otimes n}) U^{\dagger} Z^{\otimes n} \\
    &= (Z^{\otimes n}  U Z^{\otimes n}) M (Z^{\otimes n} U^{\dagger} Z^{\otimes n}) \\
    &= (Z^{\otimes n}  U Z^{\otimes n}) M (Z^{\otimes n} U Z^{\otimes n})^\dagger,
\end{align*}
where for the third equality we used the fact that conjugation by $U$ preserves parity.

This shows that $U$ and $Z^{\otimes n}  U Z^{\otimes n}$ determine the same inner automorphism (by conjugation), implying that they are equal up to a phase, that is, $Z^{\otimes n} U Z^{\otimes n} = e^{i \phi} U$ for some $\phi \in [0,2\pi)$.
Conjugating by $Z^{\otimes n}$ one obtains
\[
U = (Z^{\otimes n}Z^{\otimes n}) U (Z^{\otimes n}Z^{\otimes n}) = Z^{\otimes n}(Z^{\otimes n} U Z^{\otimes n})Z^{\otimes n} =  e^{i \phi} Z^{\otimes n}UZ^{\otimes n} = e^{2i\phi} U.
\]
Hence $2 \phi = 0 \pmod{2\pi}$,
i.e.\ $e^{i \phi} = \pm 1$ and so
$Z^{\otimes n} U Z^{\otimes n} = \pm U$.
\end{proof}

The upshot, as anticipated above, is that all gates in the matchgate hierarchy are fermionic.

\begin{corollary}
For all $n, k \geq 1$, $\GMq{n}{k} \subset \Even{n} \cup \Odd{n}$.
\end{corollary}
\begin{proof}
The Majorana operators are odd:
this can be seen from \Cref{eq:JW} and the commutation relations $ZX  = -XZ$ and $ZY = - YZ$.
The odd gates $\Odd{n}$ form a linear subspace, hence $\GMq{n}{1} \subset \Odd{n}$.
The result follows
from
\Cref{lemma:unitary_parity_preserving_conjugation_implies_evenorodd}.
\end{proof}

We write  $\Mn{k} \defeq \GMn{k} \cap \Even{n}$ for the $k$-level even gates,
The set $\Mn{2}$ comprises those unitaries realisable by matchgate circuits,
while $\GMn{2}$ corresponds to generalised matchgate circuits.

\subsection{Examples}

The $\SWAP$ gate, which acts as $\ket{x,y} \mapsto \ket{y,x}$ and can be written as $G(\one,X)$, is an example of an even two-qubit third-level gate, \ie a gate in $\Mq{2}{3}$, since
\begin{align*}
\SWAP \, c_1 \, \SWAP &= -i c_1 c_2 c_3 = J(\one,\one) 
&
\SWAP \, c_3 \, \SWAP &= -i c_1 c_3 c_4 = i J(Y, -Y) 
\\
\SWAP \, c_2 \, \SWAP &= -i c_1 c_2 c_4 = i J(-Z,Z) 
&
\SWAP \, c_4 \, \SWAP &= -i c_2 c_3 c_4 = i J(-X, X),
\end{align*} 
all of which are odd matchgates: note that $|-A|=|A|$ for any one-qubit matrix $A \in \CC^{2\times 2}$.
Another example is the controlled-$Z$ gate, $CZ = G(Z,\one) = \diag(1,1,1,-1)$, which acts as $\ket{x,y} \mapsto (-1)^{xy}\ket{x,y}$ on computational basis states.
More generally, the controlled-phase gate
$C_{2\pi / 2^{k-2}} \defeq \diag(1,1,1,e^{\frac{2 \pi i}{2^{k-2}}})$
which acts as $\ket{x,y} \mapsto e^{\frac{2 \pi i}{2^{k-2}}} \ket{x,y}$ is a two-qubit $k$-level gate --
and, in a sense made precise by \Cref{cor:2qubit_evenhierarchy_MGequi}, a paradigmatic example in $\Mq{2}{k}$,

As for three-qubit gates, here are two examples of gates in the third level, $\Mq{3}{3}$:
\[
\begin{pmatrix}
1 & 0 & 0 & 0 & 0 & 0 & 0 & 0 \\
0 & 1 & 0 & 0 & 0 & 0 & 0 & 0 \\
0 & 0 & 1 & 0 & 0 & 0 & 0 & 0 \\
0 & 0 & 0 & 1 & 0 & 0 & 0 & 0 \\
0 & 0 & 0 & 0 & 1 & 0 & 0 & 0 \\
0 & 0 & 0 & 0 & 0 & -1 & 0 & 0 \\
0 & 0 & 0 & 0 & 0 & 0 & 1 & 0 \\
0 & 0 & 0 & 0 & 0 & 0 & 0 & -1 \\
\end{pmatrix}
\qquad\qquad\qquad
\begin{pmatrix}
1 & 0 & 0 & 0 & 0 & 0 & 0 & 0 \\
0 & 0 & 0 & 0 & 1 & 0 & 0 & 0 \\
0 & 0 & 1 & 0 & 0 & 0 & 0 & 0 \\
0 & 0 & 0 & 0 & 0 & 0 & 1 & 0 \\
0 & 1 & 0 & 0 & 0 & 0 & 0 & 0 \\
0 & 0 & 0 & 0 & 0 & -1 & 0 & 0 \\
0 & 0 & 0 & 1 & 0 & 0 & 0 & 0 \\
0 & 0 & 0 & 0 & 0 & 0 & 0 & -1 \\
\end{pmatrix}
\]
On the left is a diagonal gate that picks up a phase when both the first and third modes/qubits are occupied,
acting as $\ket{x, y, z} \mapsto (-1)^{xz}\ket{x,y,z}$.
The matrix on the right corresponds to the fermionic swap (\cref{fn:fswap}) applied to the first and third qubits, $\fSWAP_{[1,3]}$,
which acts as $\ket{x,y,z} \mapsto (-1)^{xz}\ket{z,y,x}$.
Note that, even though $\fSWAP$ is a matchgate, $\fSWAP_{[1,3]}$ is not a valid matchgate circuit because the gate is not being applied to neighbouring qubits.
The assertions can be routinely verified by calculating, for each of the two matrices $U$,
the commutators $[\Lambda_3, (U c_\mu U^\dagger)^{\otimes 2}]$ for $\mu=1, \ldots, 6$ and checking that these vanish.

Another interesting family of examples at every level of the hierarchy is given by the multi-qubit controlled $Z$ gates: $CZ$, $CCZ$, $CCCZ$, \ldots.
The $n$-qubit $C^{n-1}Z$ gate is the diagonal gate on the computational basis that picks a phase of $-1$ on the computational basis state $\ket{1^n} = \ket{1 \ldots 1}$ while leaving the remaining $2^{n} - 1$ computational basis states unchanged.
We show below that this gate belongs to $\GMq{n}{n+1}$, the $(n+1)$-th level of the matchgate hierarchy.
In fact, the result applies to a more general class of diagonal gates, to which we are led by attempting to prove it: in short, these gates pick up a $-1$ phase on computational basis states that match a fixed pattern on a subset of the qubits (for the $C^{n-1}Z$ gate, the pattern is the bit string $1^n$ on \emph{all} of the $n$ qubits).
Describing these gates requires a little setting up.
Consider a vector $\vy \in \{0,1,\star\}^n$, which can be seen as a \emph{partial} $n$-bit string,
\ie a bit string where some components might not have a well-defined bit value $0$ or $1$. 
We think of $\vy$ as describing a pattern for $n$-bit strings, with $\star$ being interpreted as a `wild card'.
Formally, a bit string $\vz \in \{0,1\}^n$ \emph{matches} $\vy$, written $\vy \prec \vz$,
if for all $k \in \{1, \ldots, n\}$, $\vy_k = \star$ or $\vz_k = \vy_k$.
The \emph{length} of a pattern $\vy$, written $|\vy|$, is the number of components where it has a well-defined bit value, \ie
$|\vy| \defeq |\setdef{k \in \{1,\ldots,n\} }{\vy_k \neq \star}|$.
Note that a pattern $\vy$ is matched by precisely $2^{n-|\vy|}$ bit strings.
The pattern $\vy$ defines an $n$-qubit diagonal gate $\Fy$ which adds a $-1$ phase to the computational basis states that match it. More explicitly, $\Fy$ acts on computational basis states as $\Fy \ket{z} = (-1)^{[\vy \prec \vz]} \ket{z}$ where $[\vy \prec \vz]$ equals $1$ or $0$ depending on whether $\vy \prec \vz$ holds.

\begin{proposition}\label{prop:CnZ_level}
For any pattern vector $\vy \in \{0,1,\star\}^n$ with $|\vy|>0$, $\Fy \in \Mq{n}{|\vy|+1}$.
\end{proposition}
\begin{proof}
We first show how the $\Fy$ matrices commute past the Majorana operators:
for each $k \in \{1, \ldots, n\}$, 
\[\Fy \, c_{2k-1} \;=\; c_{2k-1} \, \Fyek \quad \text{ and } \quad \Fy \, c_{2k} \;=\; - c_{2k} \, \Fyek ,\]
where $\ek$ stands for the bit string with a single $1$ in the $k$-th position, \ie $(\ek)_j \defeq \delta_{k,j}$,
and the addition of partial bit strings is defined componentwise, with $(\vy+\ek)_j = \star$ whenever $\vy_j = \star$.
To show this, recall that $c_{2k-1} = Z_{1}\cdots Z_{k-1} X_k$ and
so it acts on a computational basis state $\ket{\vz}$ as
$c_{2k-1} \ket{\vz} = (-1)^{\vz_1 + \cdots + \vz_{k-1}}\ket{\vz+\ek}$.
Hence, for any $\vz \in \{0,1\}^n$,
\begin{align*}
\Fy \, c_{2k-1} \, \ket{\vz}
&=
 (-1)^{\vz_1 + \cdots + \vz_{k-1}} \, \Fy \, \ket{\vz+\ek}
\\
&=
(-1)^{[\vy \prec \vz +\ek ]} \, (-1)^{\vz_1 + \cdots + \vz_{k-1}} \, \ket{\vz+\ek}
\\
&=
(-1)^{[\vy \prec \vz +\ek ]} \, c_{2k-1} \, \ket{\vz}
\\
&=
(-1)^{[\vy + \ek \prec \vz ]} \, c_{2k-1} \, \ket{\vz}
\\
&= 
c_{2k-1} \, \Fyek \, \ket{\vz}
\end{align*}
where the crucial step uses the fact that $\vy \prec \vz +\ek$ if and only if $\vy + \ek \prec \vz$.
The proof for $c_{2k}$ is similar.

Secondly, we consider the product $\Fy \, \Fyek$ for a $k \in \{1, \ldots, n\}$.
We distinguish two cases.
If $\vy_k = \star$, then $\Fyek = \Fy$ and so $\Fy \, \Fyek = \Fy^2 = \one$.
Otherwise, $\Fy \, \Fyek  = \Fyslashk$ where
$\vyslashk$ is the pattern obtained from $\vy$ by `forgetting' about (\ie setting to $\star$) the $k$-th component;
formally,
\[(\vyslashk)_j = \begin{cases} \star & \text{if $j = k$} \\ \vy_j & \text{if $j \neq k$}\end{cases}.\]
This holds because each bit string $\vz$ matching $\vyslashk$ also matches exactly one of $\vy$ and $\vyek$, so that it will either `fire' $\Fy$ or $\Fyek$.
Note that $|\vyslashk| = |\vy| - 1$, and so there are double the number of bit strings matching $\vyslashk$ as there are matching $\vy$, in agreement with the previous observation.

Finally, we combine these observations to establish the result by induction on the length $|\vy|$ of the pattern.
If $|\vy| = 1$ then $\Fy = \pm Z_k =  \mp i c_{2k-1}c_{2k}$ for some $k \in \{1, \ldots, n\}$, which is a matchgate, \ie in $\Mq{n}{2}$.\footnote{For $|\vy|=0$, one has $\Fy= -\one$. According to \Cref{def:matchgate_hierarchy}, all gates in the first level are odd, so $-\one$, an even gate, is not there. But this is perhaps an indication that a more natural definition of the the hierarchy ought to include a non-empty $\Mq{n}{1}$ (which would leave the remaining levels unaffected).}
For the induction step, suppose that the result holds for patterns of length $|\vy|-1$.
Conjugation of Majorana operators by $\Fy$ yields
\[\Fy \, c_{2k-1} \, \Fy \;=\; c_{2k-1} \, \Fyek \, \Fy , \qquad  
\Fy \, c_{2k} \, \Fy \;=\; - c_{2k} \, \Fyek \, \Fy .\]
Each $\Fyek \, \Fy$ is either the identity, when $\vy_k=\star$, or it equals $\Fyslashk$.
Since $|\vyslashk| = |\vy| - 1$, by the induction hypothesis, $\Fyslashk \in \Mq{n}{|\vy|}$,
and, by \Cref{prop:GMnk-closed-phase} and \Cref{prop:GMnk-closed-leftmultiplicationMajorana} ahead,
so is the result of the conjugation of any Majorana operator.
Consequently, $\Fy \in \Mq{n}{|\vy|+1}$.
\end{proof}

\subsection{Basic properties}

We establish some basic properties of the sets of gates in the hierarchy, and relationships between them.
We first consider properties of each set of gates in the hierarchy, showing that it is closed under scaling by a phase, and under (left or right) multiplication or conjugation by Majorana operators.

\begin{proposition}\label{prop:GMnk-closed-phase}
For $n \geq 1$ and $k\geq 2$, $U \in \GMn{k}$ implies $e^{i \phi}U \in \GMn{k}$ for all $\phi \in [0,2\pi)$.
\end{proposition}
\begin{proof}
The inner automorphism determined by a unitary is unchanged by a phase.
In particular, conjugation of $c_\mu$ by $e^{i \phi}U$ yields
$(e^{i\phi}U)c_\mu(e^{i \phi}U)^\dagger = (e^{i\phi}U)c_\mu(e^{-i \phi}U^\dagger) = U c_\mu U^\dagger$.
\end{proof}

\begin{proposition}\label{prop:GMnk-closed-multiplicationMajorana}
For $n \geq 1$ and $k\geq 2$,
$U \in \GMq{n}{k}$ implies $U c_{\mu} \in \GMq{n}{k}$ for all $\mu$.
\end{proposition}
\begin{proof}
Let $U \in \GMq{n}{k}$. Conjugation of $c_{\nu}$ by $U c_{\mu}$ yields
\[
    (U c_{\mu}) c_{\nu} (U c_{\mu})^\dagger
    \;=\;
    (U c_{\mu}) c_{\nu} (c_{\mu} U^{\dagger})
    \;=\;
    U (c_{\mu} c_{\nu} c_{\mu}) U^{\dagger}
    \;=\;
    (-1)^{\delta_{\mu \nu}} U c_{\nu} U^{\dagger} 
    .
\]
Since $U c_{\nu} U^\dagger$ is in $ \GMq{n}{k-1} \cap \Odd{n}$, then so is this gate, by the definition of $\GMq{n}{1}$ (when $k=2$) or by \Cref{prop:GMnk-closed-phase} (when $k > 2$).
\end{proof}

Multiplication by a Majorana operator is an involutive operation that witnesses a bijection between $\Even{n}$ and $\Odd{n}$.
\Cref{prop:GMnk-closed-multiplicationMajorana} above shows that this restricts to a bijection between
$\Mq{n}{k} = \GMq{n}{k}\cap \Even{n}$ and $\GMq{n}{k}\cap \Odd{n}$,
the even and odd operators at each level of the hierarchy,
except for the first level which only has odd operators.

\begin{proposition}\label{prop:GMnk-closed-conjugationMajorana}
For $n, k \geq 1$,
$U \in \GMq{n}{k}$ implies $c_{\mu} U c_{\mu} \in \GMq{n}{k}$ for all $\mu$.
\end{proposition}
\begin{proof}
We establish the claim by induction on the level $k$.
For the base case $k=1$,
$U \in \GMq{n}{1}$ means that $U = \sum_{\nu}^{2n} \alpha_{\nu} c_{\nu}$ for $\alpha \in \mathbb{R}^{2n}$ with $\|\alpha\| = 1$. Conjugation of $U$ by $c_{\mu}$ yields 
\[
    c_{\mu} \left(\sum_{\nu} \alpha_{\nu} c_{\nu}\right) c_{\mu}
    \;=\;
    \sum_{\nu} \alpha_{\nu} (c_{\mu} c_{\nu} c_{\mu})
    \;=\;
    \sum_{\nu} (-1)^{\delta_{\mu \nu}}\alpha_{\nu} c_{\nu}
    \;\in\; \GMq{n}{1}.
\]
For the inductive step,
we assume that the property holds for a given $k$, \ie that $U \in \GMq{n}{k}$ implies $c_{\mu} U c_{\mu} \in \GMq{n}{k}$, and show it for $k+1$.
Given $U \in \GMq{n}{k+1}$, we consider the gate $c_{\mu} U c_{\mu}$.
Conjugation of a Majorana operator $c_{\nu}$ by $c_{\mu} U c_{\mu}$ yields:
\[
   (c_{\mu} U c_{\mu}) c_{\nu} (c_{\mu} U c_{\mu})^\dagger
   \;=\;
   (c_{\mu} U c_{\mu}) c_{\nu} (c_{\mu} U^{\dagger} c_{\mu})
   \;=\;
   c_{\mu} U (c_{\mu} c_{\nu} c_{\mu}) U^{\dagger} c_{\mu}
   \;=\;
   (-1)^{\delta_{\mu \nu}} c_{\mu}(U c_{\nu} U^{\dagger})c_{\mu}
   .
\]
Since $U c_{\nu} U^{\dagger}$ is in $\GMq{n}{k} \cap \Odd{n}$, then by the induction hypothesis and the fact that the set $\GMq{n}{k} \cap \Odd{n}$ is closed under scaling by $-1$, so is this gate.
\end{proof} 

\begin{corollary}\label{prop:GMnk-closed-leftmultiplicationMajorana}
For $n \geq 1$ and $k\geq 2$,
$U \in \GMq{n}{k}$ implies $c_{\mu} U\in \GMq{n}{k}$ for all $\mu$.
\end{corollary}
\begin{proof}
Chain \Cref{prop:GMnk-closed-conjugationMajorana,prop:GMnk-closed-multiplicationMajorana}
to obtain
$U \in \GMq{n}{k}$ implies $c_\mu U c_\mu \in \GMq{n}{k}$ implies $(c_\mu U c_\mu) c_\mu \in \GMq{n}{k}$,
and observe $(c_\mu U c_\mu) c_\mu = c_\mu U c_\mu^2 = c_\mu U$.
\end{proof}

Moving on to properties relating different sets in the hierarchy,
we show that for a fixed number of qubits, 
the levels of the hierarchy are nested.
\begin{proposition}\label{prop:nested}
For $n, k \geq 1$, $\GMq{n}{k} \subset \GMn{k+1}$.
\end{proposition}
\begin{proof}
It is enough to show that $\GMq{n}{1} \subset \GMq{n}{2}$, and the remaining inclusions then follow by induction on the level.
Consider the conjugation of a Majorana operator $c_\mu$ by a first-level gate $\sum_\lambda a_\lambda c_\lambda$ with $\sum_\lambda a_\lambda^2=1$:
\[
    \left(\sum_\lambda a_\lambda c_\lambda\right)c_\mu\left(\sum_\nu a_\nu c_\nu\right)^\dagger = \sum_{\lambda,\nu}a_\lambda a_\nu c_\lambda c_\mu c_\nu .
\]
We show that this is itself in $\GMq{n}{1} \cap \Odd{n} = \GMq{n}{1}$.
We break the sum on the right-hand side into four parts,
 according to whether the indices $\lambda, \nu$ match $\mu$:
 \[\lambda=\mu \text{ and } \nu = \mu, \qquad \lambda=\mu \text{ and } \nu \neq \mu, \qquad \lambda\neq\mu \text{ and } \nu = \mu, \qquad \lambda\neq\mu \text{ and } \nu \neq \mu.\] 
Using $c_\mu^2=\one$, this yields:
\[
    a_\mu^2 c_\mu
    +
    \sum_{\nu \neq \mu} a_\mu a_\nu c_\nu
    +
    \sum_{\lambda \neq \mu} a_\lambda a_\mu c_\lambda
    +
    \sum_{\lambda, \nu \neq \mu} a_\lambda a_\nu c_\lambda c_\mu c_\nu.
\]
The middle terms are identical while the last can be rewritten as
\[
\sum_{\lambda, \nu \neq \mu} a_\lambda a_\nu c_\lambda c_\mu c_\nu
\;=\;
- \sum_{\lambda, \nu \neq \mu} a_\lambda a_\nu c_\mu c_\lambda c_\nu 
\;=\;
- c_\mu \big(\sum_{\lambda \neq \mu}a_\lambda c_\lambda\big)^2  
\;=\;
-c_\mu \sum_{\lambda \neq \mu} a_\lambda^2 
\;=\;
(a_\mu^2-1)c_\mu
\]
using anticommutation of $c_\mu$, $c_\lambda$ given $\lambda\neq\mu$, then \Cref{eq:squarelinci}, and finally the fact that the vector of coefficients
for the first level gate has unit norm, \ie\ $\sum_\lambda a_\lambda^2=1$.
Combining the observations above, we obtain:
\begin{equation}\label{eq:conjugation_c_mu_by_first_level}
    \left(\sum_\lambda a_\lambda c_\lambda\right)c_\mu\left(\sum_\nu a_\nu c_\nu\right)^\dagger 
    \;=\; (2a_\mu^2 - 1) c_\mu + \sum_{\lambda \neq \mu} (2 a_\mu a_\lambda) c_\lambda,
\end{equation}
which is a linear combination of Majorana operators.
It remains to check that its coefficients are normalised:
\begin{align*}
    (2 a_\mu^2 - 1)^2 + \sum_{\lambda \neq \mu}(2a_\mu  a_\lambda)^2
    \;=\;
    1 - 4a_\mu^2 + 4a_\mu^2 a_\mu^2  + 4a_\mu^2 \sum_{\lambda\neq \mu} a_\lambda^2
    \;=\;
    1 + 4a_\mu^2\big(-1 + \overbrace{a_\mu^2 + \sum_{\lambda\neq \mu} a_\lambda^2}^1\big)
    \;=\; 1. 
\end{align*}
\end{proof}

\begin{remark}
From \Cref{eq:conjugation_c_mu_by_first_level} in the proof above one can derive the general  form of the action of $\GMq{n}{1}$ on itself by conjugation: 
\[
\left(\sum_\lambda a_\lambda c_\lambda\right)\left(\sum_\mu b_\mu c_\mu\right)\left(\sum_\nu a_\nu c_\nu\right)^\dagger 
=
\sum_{\mu,\lambda} 2 a_\mu b_\mu a_\lambda c_\lambda  -  \sum_{\mu}b_\mu c_\mu 
=
2  \left(\sum_{\lambda}  a_\lambda c_\lambda\right) \left(\sum_\mu a_\mu b_\mu\right)  -  \sum_{\mu}b_\mu c_\mu  .
\]
In terms of the coefficient vectors in $\RR^{2n}$ that determine the first-level gates,
the action by conjugation of (the first-level gate determined by vector) 
$\vec{a} \in \RR^{2n}$ is represented by the linear map $\RR^{2n} \to \RR^{2n}$ 
mapping $\vec{b} \in \RR^{2n}$ to 
\[2  \vec{a} \langle \vec{a},\vec{b}\rangle  - \vec{b}
=
2  \vec{a} (\vec{a}^\intercal \vec{b})  - \vec{b}
=
(2  \vec{a} \vec{a}^\intercal - \one) \vec{b}.
\]
This map $2  \vec{a} \vec{a}^\intercal - \one$ is a reflection  through $\vec{a}$ in $\RR^{2n}$.
\end{remark}

We now consider the relation between levels of the hierarchy across different number of qubits.
We first observe how the Majorana operators decompose in terms of Majoranas for fewer fermionic modes (or qubits).
Given $m, n \in \NN$, the $(n+m)$-mode Majorana operators $\{c^{(n+m)}_\mu\}_{\mu=1}^{2(n+m)}$ act on $n+m$ qubits
and have the following form:
\[
c^{(n+m)}_\mu =
\begin{cases}
c^{(n)}_\nu \otimes \one^{\otimes m} & \text{\small if $1 \leq \mu \leq 2n$ putting $\nu = \mu$ }
\\
Z^{\otimes n} \otimes c^{(m)}_\nu  & \text{\small if $2n+1 \leq \mu \leq 2n+2m$ putting $\nu = \mu-2n$ .}
\end{cases} \]
We now show that the levels of the hierarchy (except the first) are closed under tensor product.
\begin{proposition} \label{prop: U_tensor_V_k_level}
For $n, m \geq 1$ and $k \geq 2$, $U \in \GMq{n}{k}$ and $V \in \GMq{m}{k}$ implies $U \otimes V \in \GMq{n+m}{k}$.
\end{proposition}
\begin{proof}
Consider the conjugation of the $(n+m)$-qubit Majorana operators by the matrix $U \otimes V \in \U(2^{n+m})$. 
For a Majorana of the form $c^{(n)}_\nu \otimes \one^{\otimes m}$, it yields
\begin{equation}\label{eq:conj_UtensorV_1}
(U \otimes V)(c^{(n)}_\nu \otimes \one^{\otimes m})(U^\dagger \otimes V^\dagger)
\;=\;
U c^{(n)}_\nu U^\dagger \otimes V\one^{\otimes m}V^\dagger
\;=\;
Uc^{(n)}_\nu U^\dagger \otimes \one^{\otimes m} ;
\end{equation}
while for one of the form $Z^{\otimes n} \otimes c^{(m)}_\nu$,
it yields
\begin{equation}\label{eq:conj_UtensorV_2}
(U \otimes V)(Z^{\otimes n}\otimes c^{(m)}_\nu)(U^\dagger \otimes V^\dagger) 
\;=\;
 UZ^{\otimes n}U^\dagger \otimes V c^{(m)}_\nu V^\dagger 
\;=\;
\pm Z^{\otimes n} \otimes V c^{(m)}_\nu V^\dagger,
\end{equation}
where we use the fact that $U$ is fermionic (even or odd).

Now, if $U$ and $V$ are $(k+1)$-level gates for some $k \geq 1$,
then \Cref{eq:conj_UtensorV_1}
yields
a unitary in $(\GMq{n}{k} \cap \Odd{n}) \otimes \one^{\otimes m}$,
an odd $n$-qubit $k$-level gate tensored with the $m$-qubit identity on the right, 
while \Cref{eq:conj_UtensorV_2}
yields
a gate in  $Z^{\otimes n} \otimes (\GMq{m}{k-1} \cap \Odd{m})$,
an odd $m$-qubit $k$-level gate tensored with the $n$-qubit parity operator on the left.
Since both $\one^{\otimes m}$ and $Z^{\otimes n}$ are even, the gates in \Cref{eq:conj_UtensorV_1,eq:conj_UtensorV_2} are both odd, \ie in $\Odd{n+m}$.

We establish the main claim by induction on the level.
For the base case, $k+1=2$, we use the fact that tensoring with parity on the left or identity on the right maps Majorana operators to Majorana operators, and thus the gates in \Cref{eq:conj_UtensorV_1,eq:conj_UtensorV_2} are in $\GMq{n+m}{1}$.
For the inductive step, we assume that the property holds for a level $k$, with $k\geq 2$, and show it for $k+1$.
Observe that both $\one^{\otimes n}$ and $Z^{\otimes n}$
are in $\GMq{n}{2}$ -- and thus also in $\GMq{n}{k}$ by \Cref{prop:nested} -- as they induce by conjugation the linear actions $c_\mu \mapsto c_\mu$ and $c_\mu \mapsto - c_\mu$, respectively.
One can thus apply the induction hypothesis to conclude that the gates in \Cref{eq:conj_UtensorV_1,eq:conj_UtensorV_2} are in $\GMq{n+m}{k}$, as required.
\end{proof}

\section{Hierarchy gate teleportation using magic states}\label{sec:teleportation_protocol}

Inspired by the gate teleportation construction proposed in Ref.~\cite{Gottesman_1999} for the deterministic implementation of \emph{any}
gate in the Clifford hierarchy,
we construct an analogous gate teleportation protocol for implementing gates in the matchgate hierarchy using matchgate circuits.
We first present the protocol for the case of two-qubit gates and then its generalisation to an arbitrary number of qubits.

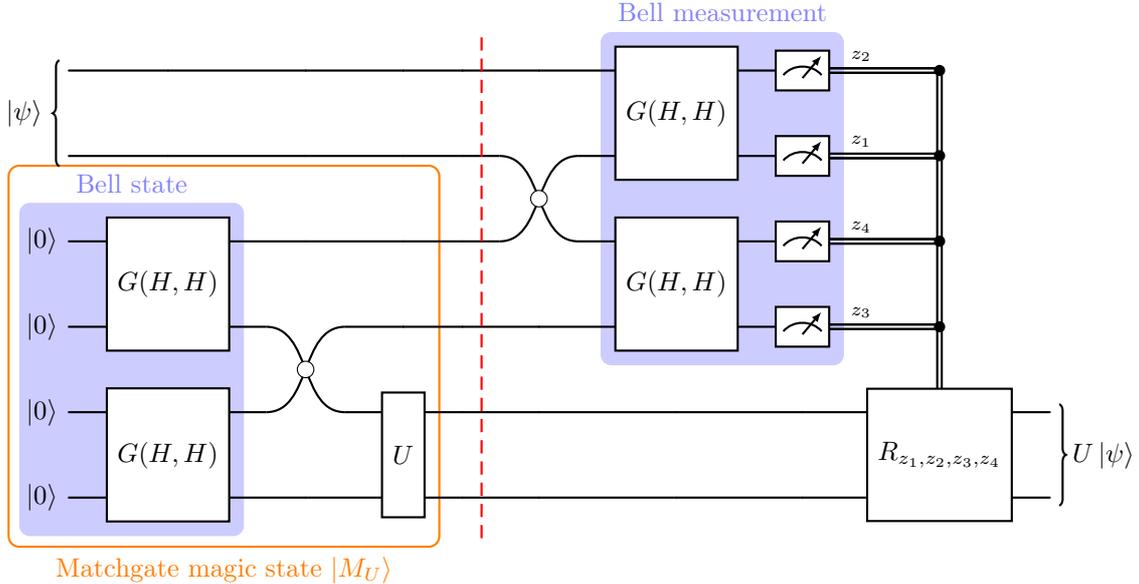
\begin{figure}[!ht]
\centering
\tikzset{my label/.append style={above right,xshift=0.45cm,yshift=-0.25cm}}
\begin{quantikz}
\lstick[2]{$\ket{\psi}$} &  &  &  & \slice{} &  & \gate[2]{\textsc{$G(H, H)$}}
\gategroup[4,steps=2,style={draw=none,
rounded corners,fill=blue!20, inner xsep=2pt, inner ysep=2pt},
background]{{\textcolor{blue!50}{Bell measurement}}} & \meter{} &  \ctrl[vertical wire=c]{1} \setwiretype{c}\wire[l][1]["z_2"{above,pos=0.7}]{a} 
\\
& & & & &
\gate[2,fswap,style={color=white}]{}  
&  & \meter{} & \ctrl[vertical wire=c]{1} \setwiretype{c}\wire[l][1]["z_1"{above,pos=0.7}]{a}  
\\
\gategroup[4,steps=4,
            style={color=orange, rounded corners,fill=none, inner xsep=10.5pt, xshift=-8.5pt, inner ysep=11pt, yshift=5pt},
            background,
            label style={label position=below,anchor=north,yshift=-0.2cm}]{{\textcolor{amber}{Matchgate magic state $\ket{M_U}$}}}
\gategroup[4,steps=2,
            style={draw=none, rounded corners,  fill=blue!20, inner xsep=8.5pt, xshift=-6.5pt, inner ysep=2pt},
            background]{{\textcolor{blue!50}{Bell state}}}
\lstick{$\ket{0}$}
& \gate[2]{\textsc{$G(H, H)$}}
&  &  &  &  & \gate[2]{\textsc{$G(H, H)$}} & \meter{} & \ctrl[vertical wire=c]{1} \setwiretype{c}\wire[l][1]["z_4"{above,pos=0.7}]{a} 
\\
\lstick{$\ket{0}$} & &
\gate[2,fswap,style={color=white}]{}  
& & & &  & \meter{} &   \ctrl[vertical wire=c]{1} \setwiretype{c}\wire[l][1]["z_3"{above,pos=0.7}]{a}  \\
\lstick{$\ket{0}$} &  \gate[2]{\textsc{$G(H, H)$}} &  & \gate[2]{U} & & & & & \gate[wires=2]{R_{z_1,z_2, z_3, z_4}} & \rstick[wires=2]{$U \ket{\psi}$}  \\
\lstick{$\ket{0}$} & & & & & & & & & 
\end{quantikz}
\caption[Circuit to implement a gate $U \in \GMn{k+1}$ (here depicted for $n=2$) by quantum gate teleportation]{Circuit to implement a gate $U \in \GMn{k+1}$ (here depicted for $n=2$) by quantum gate teleportation,
using only nearest-neighbour matchgates
$\fSWAP$ 
(depicted as \raisebox{1.8pt}{\resizebox{!}{11pt}{\begin{quantikz}\gate[2,fswap,style={color=white}]{}\\ \end{quantikz}}})
and $G(H,H)$, a pre-prepared resource `matchgate-magic' state, and `correction'  operators $R_{\vz} \defeq i^{\alpha(z)} U(\prod_{\mu} c_\mu^{z_\mu})U^\dagger$ indexed by measurement outcomes $\vz \in \{0,1\}^{2n}$, which are sequences of gates in the lower level $\GMn{k}$ of the hierarchy.
For $n=2$, this sequence is a single gate in $\GMq{2}{k}$ since all levels are groups; see \Cref{rem: efficiency_2_qubit_protocol} ahead following the characterisation of the gates in $\GMq{2}{k}$.
}
\label{fig:teleportation_protocol}
\end{figure}

\subsection{Protocol for two-qubit gates}
The proposed protocol is depicted in \Cref{fig:teleportation_protocol} for the two-qubit gate case.
Its effect is to \emph{deterministically} implement a two-qubit gate $U$ from the hierarchy, applying it to the input received on the top two wires. The protocol consumes a four-qubit matchgate-magic state
\begin{equation}\label{eq:MU}
    \ket{M_U} \;=\; (\one\otimes U) \, B \ket{0000} 
\end{equation}
where $U$ is the target gate and 
\begin{align*}
B
& \;=\; \fSWAP_{[2,3]} G(H,H)_{[1,2]} G(H,H)_{[3,4]} \\
& \;=\; (\one \otimes \fSWAP \otimes \one) \, (G(H,H) \otimes G(H,H)) ,
\end{align*}
where we recall from \cref{fn:fswap} that $\fSWAP = G(Z,X)$ is the fermionic swap.

The protocol proceeds by measuring four of the qubits: the two input wires and one half of the magic state.
The four qubits are measured in the computational basis after the application of the matchgate circuit $B^\dagger$.
When all the obtained outcomes equal zero, we know that the gate $U$ has been implemented.
So far, this gives a method for \emph{probabilistically} implementing any unitary gate $U$ (not necessarily in the hierarchy).
As with the analogous protocol using Clifford circuits, the need for post-selection renders the approach impractical as the probability of success drops exponentially with the number of non-matchgates to be implemented in a given circuit.

The fact that $U$ is in the hierarchy guarantees that the output can be \emph{corrected} based on the obtained outcomes.
The necessary correction is, in general, a sequence of gates from a lower level in the matchgate hierarchy.\footnote{For the case of two-qubit gates, at least, from the results in \Cref{sec:2qubit_characterisation} ahead, we know that this correction sequence actually reduces to a single gate. This guarantees that resource requirements for implementing a two-qubit gate grow linearly with its level in the hierarchy; see \Cref{rem: efficiency_2_qubit_protocol}.}
Each of these gates is either a (directly implementable) matchgate or it can in turn be implemented by another instance of the same protocol. 
Therefore, recursive application of the protocol allows one to \emph{deterministically} implement \emph{any} gate in the matchgate hierarchy
using only matchgate circuits and a supply of (pre-prepared) matchgate `magic' states. 

We thus also provide a whole family of \emph{deterministic} magic states for matchgate computation,
greatly generalising the one known deterministic matchgate-magic state used to implement the $\SWAP$ gate \cite{hebenstreit_2019}.

\begin{remark}[Using matchgate magic states]\label{rem:using_matchgate_magic}
Our protocol uses a resource state $\ket{M_U}$ from \Cref{eq:MU} in order to implement a gate $U$.
Despite not being Gaussian states, meaning that they cannot be prepared using a matchgate circuit, such
resource states can be prepared offline beforehand, \ie independently of the input to the circuit,
and even independently of the specific circuit itself that is being implemented.
In using such matchgate magic states, a crucial aspect to consider is that matchgate circuits are built out of gates acting on neighbouring qubit lines, with the $\SWAP$ gate not being freely available.
This places a singular emphasis on the geometric arrangement of qubits,\footnote{See Ref.~\cite{Brod_2012} for a discussion of computational models with nearest-neighbour matchgate interactions among qubits laid out in alternative geometric arrangements to the usual one-dimensional chain, most of which achieve quantum universality.}
as qubit states cannot in general be freely moved around the qubit lines to undergo an interaction with a distant qubit.
In order to use our gate teleportation protocol for implementing a gate $U$ -- which may be required at an arbitrary place within a circuit --,
the preprepared magic state $\ket{M_U}$ must be moved into position next to the qubit lines where the gate is to be implemented (and likely also moved back out to be discarded).
For that, one must be able to swap the magic state $\ket{M_U}$ through arbitrary states via the use of free gates only, \ie by a matchgate circuit.
Hebenstreit et al.\ \cite[Lemma 1]{hebenstreit_2019} showed that this can be done if and only if the magic state is fermionic.
In turn, our magic state $\ket{M_U}$ from \Cref{eq:MU} is fermionic when the gate $U$ is itself fermionic, since
\begin{align*}
Z^{\otimes 2n} \ket{M_U}
&\;=\; (Z^{\otimes n} \otimes Z^{\otimes n}) (\one \otimes U) B \ket{0000}
\;=\; (Z^{\otimes n} \otimes Z^{\otimes n} U) B \ket{0000}
\\
&\;=\; (Z^{\otimes n} \otimes \pm U Z^{\otimes n}) B \ket{0000}
\;=\; \pm (\one \otimes U) (Z^{\otimes n} \otimes Z^{\otimes n}) B \ket{0000}
\\
&\;=\; \pm (\one \otimes U) B Z^{\otimes 2n} \ket{0000}
\;=\; \pm (\one \otimes U) B \ket{0000}
\\
&\;=\; \pm \ket{M_U}
\end{align*}
where the plus or minus depends on whether the gate $U$ is even or odd, respectively.
The more structural reason why this works is that fermionic maps are closed under composition and tensor products.
As a consequence, the protocol can be used to implement any circuit using gates from $\GMq{n}{k}$ in any level $k$ of the hierarchy.
\end{remark}

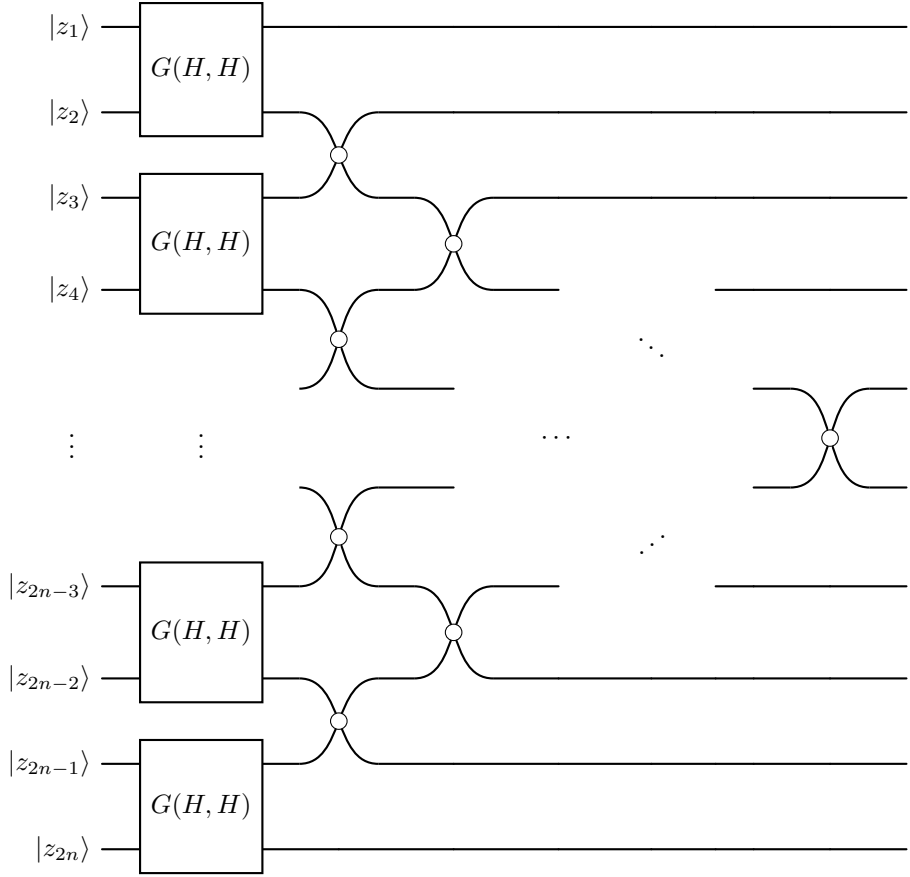
\begin{figure}[h!]
\centering
\tikzset{my label/.append style={above right,xshift=0.45cm,yshift=-0.25cm}}
\begin{quantikz}
\lstick{$\ket{z_1}$} & \gate[2]{\textsc{$G(H, H)$}} &  &  & &  & & & &  
\\
\lstick{$\ket{z_2}$} & & \gate[2,fswap,style={color=white}]{}   & &  &  & & & &  
\\
\lstick{$\ket{z_3}$} &  \gate[2]{\textsc{$G(H, H)$}} &  &
\gate[2,fswap,style={color=white}]{}  
&  & & & & & 
\\
\lstick{$\ket{z_4}$} & & \gate[2,fswap,style={color=white}]{}  
 & & &  \gate[2,style={fill=white,draw=white}]{\ddots}
 \setwiretype{n}
 & & \setwiretype{q} & & 
\\
\setwiretype{n}
\lstick[2,brackets=none]{$\vdots$} 
&  \gate[2,style={fill=white,draw=white}]{\vdots} 
&  
&  \setwiretype{q}& \gate[2,style={fill=white,draw=white}]{\cdots}\setwiretype{n}
&  
& &  & \gate[2,fswap,style={color=white}]{} \setwiretype{q} 
 &  
\\
\setwiretype{n}
& 
& \gate[2,fswap,style={color=white}]{}
& \setwiretype{q} & \setwiretype{n} & 
\gate[2,style={fill=white,draw=white}]{\iddots} &  & & \setwiretype{q}  &  
\\
\lstick{$\ket{z_{2n-3}}$} &  \gate[2]{\textsc{$G(H, H)$}} & \swap{-1} & \gate[2,fswap,style={color=white}]{} &  &  
\setwiretype{n} & & \setwiretype{q} & & 
\\
\lstick{$\ket{z_{2n-2}}$} & & \gate[2,fswap,style={color=white}]{}   &   & &  & & & & 
\\
\lstick{$\ket{z_{2n-1}}$} &  \gate[2]{\textsc{$G(H, H)$}} &  & & & & & & & 
\\
\lstick{$\ket{z_{2n}}$} & &  &  & & & & & & 
\end{quantikz}
\caption{Gate $B^{(n)}$ used in the $n$-qubit protocol.}
\label{fig:teleportation_protocol_N}
\end{figure}

\subsection{Protocol for \texorpdfstring{$n$}{n}-qubit gates}
The protocol generalises to implement matchgate hierarchy gates over an arbitrary number $n\geq 2$ of qubits.
The occurrences of the four-qubit unitary $B$ and of its adjoint in the circuit of \Cref{fig:teleportation_protocol} are replaced by the matchgate circuit over $2n$ qubits
shown in \Cref{fig:teleportation_protocol_N} and its adjoint:
\begin{equation}\label{eq:Bn}
B^{(n)} \;=\; S^{(n)} \, \left(\prod_{k=1}^n G(H,H)_{[2k-1,2k]}\right)
\end{equation}
where
\begin{align*}
   S^{(n)} \; = \; &  \left(\fSWAP_{[n,n+1]}\right)
   \\        &  \left(\fSWAP_{[n-1,n]} \, \fSWAP_{[n+1,n+2]}\right)
   \\ &  \; \cdots
   \\ &  \left(\fSWAP_{[3,4]} \, \cdots \, \fSWAP_{[2n-3,2n-2]}\right)
   \\ &  \left(\fSWAP_{[2,3]} \, \cdots \, \fSWAP_{[2n-2,2n-1]}\right) .
\end{align*}
This matchgate circuit $B^{(n)}$
is built out of a layer of $G(H,H)$ gates applied on consecutive pairs,
followed by the circuit $S^{(n)}$ consisting of $n-1$ layers of fermionic swap gates, $\fSWAP = G(Z,X)$, arranged in a triangle so that
each such layer has one fewer gate than the preceding one, 
for a total of 
${n(n-1)}/{2}$ fermionic swaps.
In essence, the $S^{(n)}$ part of the circuit implements a permutation of the wires whereby odd-numbered qubits get listed first (in the original order) and are then followed by even-numbered qubits (also in their original order). This permutation is depicted in \Cref{fig:fswapmultiple}.
However, $S^{(n)}$ is not an ordinary permutation of qubits,
since the $\SWAP$ gate is not available in matchgate circuits. Instead, each transposition of neighbouring qubits is implemented by a fermionic swap gate, which picks up a $-1$ phase whenever both modes being swapped are occupied, \ie on the computational basis state $\ket{11}$.

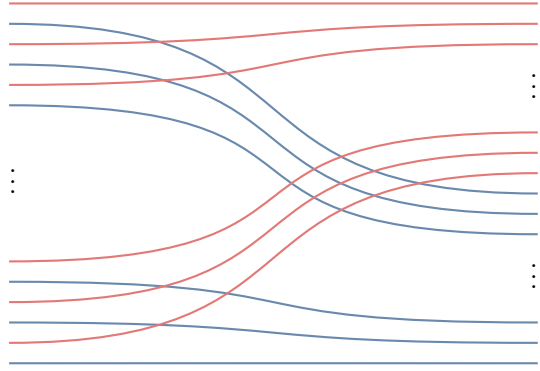
\begin{figure}[htb]
\definecolor{col1}{HTML}{E57A77} 
\definecolor{col2}{HTML}{7CA1CC} 
\centering
\begin{tikzpicture}[scale=1.8]
\centering
\def\step{-.15}
\def\lstep{-.5}
\def\wid{4}
\def\loo{1.7}
\centering
\path      (0,0)    node (c1) {}
        ++(0,\step) node (c2) {}
        ++(0,\step) node (c3) {}
        ++(0,\step) node (c4) {}
        ++(0,\step) node (c5) {}
        ++(0,\step) node (c6) {}
        ++(0,\lstep) node[right] (mid) {$\vdots$}
        ++(0,\lstep+\step) node (cm6) {}
        ++(0,\step) node (cm5) {}
        ++(0,\step) node (cm4) {}
        ++(0,\step) node (cm3) {}
        ++(0,\step) node (cm2) {}
        ++(0,\step) node (cm1) {};

        \path (\wid,0) node (u1) {}
        ++(0,\step) node (u2) {}
        ++(0,\step) node (u3) {}
        ++(0,0.5*\lstep) node[left] (midup) {$\vdots$}
        ++(0,0.5*\lstep+\step) node (um3) {}
        ++(0,\step) node (um2) {}
        ++(0,\step) node (um1) {}
        ++(0,\step) node  (d1) {}
        ++(0,\step) node (d2) {}
        ++(0,\step) node (d3) {}
        ++(0,0.5*\lstep) node[left] (midup) {$\vdots$}
        ++(0,0.5*\lstep+\step) node (dm3) {}
        ++(0,\step) node (dm2) {}
        ++(0,\step) node (dm1) {};

\draw[thickness=2,color=col1!99!black]  (c1)  to[out=0,in=180,looseness=0]  (u1);
\draw[thickness=2,color=col2!85!black]  (c2)  to[out=0,in=180,looseness=\loo-0.2]  (d1);
\draw[thickness=2,color=col1!99!black]  (c3)  to[out=0,in=180,looseness=\loo-0.2]  (u2);
\draw[thickness=2,color=col2!85!black]  (c4)  to[out=0,in=180,looseness=\loo-0.1]  (d2);
\draw[thickness=2,color=col1!99!black]  (c5)  to[out=0,in=180,looseness=\loo-0.1]  (u3);
\draw[thickness=2,color=col2!85!black]  (c6)  to[out=0,in=180,looseness=\loo]  (d3);
\draw[thickness=2,color=col1!99!black]  (cm6) to[out=0,in=180,looseness=\loo]  (um3);
\draw[thickness=2,color=col2!85!black]  (cm5) to[out=0,in=180,looseness=\loo-0.1]  (dm3);
\draw[thickness=2,color=col1!99!black]  (cm4) to[out=0,in=180,looseness=\loo-0.1]  (um2);
\draw[thickness=2,color=col2!85!black]  (cm3) to[out=0,in=180,looseness=\loo-0.2]  (dm2);
\draw[thickness=2,color=col1!99!black]  (cm2) to[out=0,in=180,looseness=\loo-0.2]  (um1);
\draw[thickness=2,color=col2!85!black]  (cm1) to[out=0,in=180,looseness=0]  (dm1);
\end{tikzpicture}
\caption{Permutation of wires determined by the fermionic swaps in $S^{(n)}$.}\label{fig:fswapmultiple}
\end{figure}

\subsection{Verification of the protocol}
To show correctness of this protocol, we consider the action of each gate on the computational basis states.
We perform these calculations for the case of two-qubit gates, but the proof for arbitrary number of qubits is similar.

The matchgate $G(H,H)$ acts as follows on a computational basis state $\ket{x,y}$:
\begin{equation}
    G(H,H)\ket{x,y} = \frac{\ket{0,x+y} + (-1)^x\ket{1, x+y+1}}{\sqrt{2}}
 = \frac{\sum_a (-1)^{ax}\ket{a, x+y+a}}{\sqrt{2}}
\end{equation}
whereas the action of the fermionic swap gate is
\begin{equation}
    \fSWAP\ket{x,y} = (-1)^{xy}\ket{y,x}.
\end{equation}

The action of the matchgate circuit $B=\fSWAP_{23}G(H,H)_{12}G(H,H)_{34}$
in a general computational basis state $\ket{z_1 z_2 z_3 z_4}$ is then:
\begin{align*}
    &\fSWAP_{[2,3]}G(H,H)_{[1,2]}G(H,H)_{[3,4]}\ket{z_1, z_2, z_3, z_4}
    \\
    = \; & \fSWAP_{[2,3]}\left(\frac{\sum_{a} (-1)^{a z_1}\ket{a, z_1 + z_2 +a}}{\sqrt{2}}\right)\left(\frac{\sum_{b} (-1)^{b z_3}\ket{b, z_3 + z_4 + b}}{\sqrt{2}}\right)
    \\
   = \; &\frac{1}{2} \, \fSWAP_{[2,3]}\left(\sum_{ab}(-1)^{a z_1 + b z_3}\ket{a, z_1 + z_2 +a, b, z_3 + z_4 + b}\right)
   \\
   = \; &\frac{1}{2}\sum_{ab}(-1)^{a z_1 + b z_3 + b (z_1 + z_2 + a)}\ket{a ,b, z_1 + z_2 +a, z_3 + z_4 + b}
   \\
   = \; &\frac{1}{2}\sum_{ab}(-1)^{a z_1 + b z_3 + b z_1 + b z_2 + ab}\ket{a ,b, z_1 + z_2 +a, z_3 + z_4 + b}
\end{align*}

Now consider the whole protocol for the particular case of $U=\one^{(2)}$, the identity gate on two qubits:
\begin{align}\label{eq:protocol_snake}
    &
    (\bra{z_1,z_2,z_3,z_4} \otimes \one^{(2)}) (B^\dagger \otimes \one^{(2)})( \one^{(2)} \otimes B) (\ket{x,y} \otimes \ket{0,0,0,0})
    \notag
    \\
    = \; &
    \frac{1}{4}\left(\sum_{ab}(-1)^{a z_1 + b (z_1+z_2+z_3) + ab}\bra{a, b, a+z_1 + z_2 , b + z_3 + z_4 } \otimes  \one^{(2)}\right)
    \left(\ket{x,y} \otimes \sum_{cd}(-1)^{cd}\ket{c, d, c, d}\right)
    \notag
    \\
    = \; &
    \frac{1}{4}\left(\sum_{ab}(-1)^{a z_1 + b (z_1+z_2+z_3) + ab}\bra{a, b, a+z_1 + z_2 , b + z_3 + z_4 } \otimes  \one^{(2)}\right)
    \left(\sum_{cd}(-1)^{cd}\ket{x,y,c, d, c, d}\right)
    \notag
    \\
    = \; &
    \frac{1}{4}\sum_{abcd}(-1)^{a z_1 + b (z_1+z_2+z_3) + ab+cd}\left(\bra{a, b, a+z_1 + z_2 , b + z_3 + z_4} \otimes \one^{(2)} \right) \ket{x,y,c, d, c, d}
    \notag
    \\ 
    = \; &
    \frac{1}{4}(-1)^{x z_1 + y (z_1+z_2+z_3) + x y+(x+z_1+z_2)(y+z_3+z_4)}\ket{x+z_1+z_2, y + z_3 + z_4}
    \notag
    \\ 
    = \; &
    \frac{1}{4}(-1)^{x z_1 + y z_3 + x (z_3+z_4) + (z_1+z_2)(z_3+z_4)}\ket{x+z_1+z_2, y + z_3 + z_4},
    \end{align}
where for the fourth equality we used the identifications
$a=x$, $b=y$, $c=a+z_1+z_2$, and $d = b + z_3+z_4$,
which label the only non-zero term in the preceding sum.

On the other hand, the action of the Majorana operators for two qubits is given as follows on computational basis states $\ket{x,y}$:
\begin{align*}
    c^{z_2}_1 \ket{x,y} &= \ket{x  + z_2, y}
    &c^{z_1}_2 \ket{x,y} &= i^{z_1}(-1)^{x z_1}\ket{x + z_1, y} \\
    c^{z_4}_3 \ket{x,y} &= (-1)^{x z_4}\ket{x, y + z_4}
    & c^{z_3}_4 \ket{x,y} &= i^{z_3}(-1)^{x z_3 + y z_3}\ket{x, y + z_3}
\end{align*}
The product 
 $c^{z_2}_1 c^{z_1}_2 c^{z_4}_3 c^{z_3}_4$ then yields:
\[
c^{z_2}_1 c^{z_1}_2 c^{z_4}_3 c^{z_3}_4 \ket{x,y} = i^{z_1+z_3}(-1)^{xz_1+yz_3}(-1)^{x(z_3 + z_4)} \ket{x + z_1 + z_2, y + z_3 + z_4}.
\]
The output of the circuit in \Cref{eq:protocol_snake},
post-selected on the observation of classical outcomes  $z_1, z_2, z_3, z_4$, is
\begin{equation} \label{correction_2}
   (-i)^{z_1 + z_3}\\(-1)^{(z_1 + z_2)(z_3 + z_4)} c^{z_2}_1 c^{z_1}_2 c^{z_4}_3 c^{z_3}_4 \ket{x,y}. 
\end{equation}

More generally, for $n$ qubits, we replace $B$ by the unitary $B^{(n)}$ in \Cref{eq:Bn}.
We can use \Cref{fig:fswapmultiple} to visualise the fermionic crossings and track the phases.
Measurement outcomes $\vz \in\{0,1\}^{2n}$ witness the implementation of the gate
\begin{multline} \label{correction_N}
    (-i)^{z_1 +  z_3 + \cdots + z_{2n-1}}
    (-1)^{(z_1 + z_2)(z_3 + \cdots+z_{2n})+(z_3+z_4)(z_5+\cdots+z_{2n}) + \cdots + (z_{2n-3}+z_{2n-2})(z_{2n-1}+z_{2n})}
    \\
    c_1^{z_2} c_2^{z_1} \ldots c_{2n-1}^{z_{2n}} c_{2n}^{z_{2n-1}}.
\end{multline}
In order to accomplish ordinary teleportation, \ie for $U$ equal to the identity matrix, it is necessary to perform the correction to undo the gate in \Cref{correction_N}. 

When the goal is to implement a general $n$-qubit $U$, one uses the $2n$-qubit magic state  $\ket{M_U} = (\one\otimes U)B^{(n)}\ket{0}^{\otimes {2n}}$,
and the required post-measurement correction is then the conjugation of the correction from \Cref{correction_N} by the unitary $U$.
If $U$ is in level $k+1$ of the matchgate hierarchy,
then the corrections can be implemented by a sequence of $k$-level gates.
Iterating this procedure, one eventually needs only generalised matchgate circuits with adaptive measurements to \emph{deterministically}
implement any gate in the hierarchy.

\section{Characterisation of two-qubit gates in the hierarchy}\label{sec:2qubit_characterisation}

\subsection{General form}

In the two-qubit case, we obtain a neat closed-form characterisation of the
gates at each level of the matchgate hierarchy.
Recall that all the gates in the hierarchy are fermionic, \ie either odd or even.
All even two-qubit gates are of the form $G(A,B)$ for unitaries $A, B \in \U(2)$; see \Cref{eq:G_AB}.
Similarly, all odd two-qubit gates are of the form 
\begin{equation}\label{eq:J_AB}
J(A,B) \defeq \begin{pmatrix}
0 & A_{11} & A_{12} & 0 \\
B_{11} & 0 &  0 &  B_{12} \\
B_{21} & 0 & 0  & B_{22}  \\
0 & A_{21} & A_{22} & 0 
\end{pmatrix}
\end{equation}
for unitaries $A, B \in \U(2)$.
Note that these are more general than two-qubit (generalised) matchgates in that there is no restriction that the determinants of $A$ and $B$ be equal.
The two-qubit gates at each level of the hierarchy are characterised precisely
in terms of the relation between the determinants of their components $A$ and $B$.

We start by giving a general form for the two-qubit gates in the first level of the hierarchy.

\begin{proposition}\label{prop:2qubit_firstlevel_characterisation}
$\GMq{2}{1}$ is the set of unitaries of the form $J(A,A^\dagger)$ with $|A|=-1$.
\end{proposition}
\begin{proof}
By definition, a general gate in $\GMq{2}{1}$ has the form
$a_1 c_1 + a_2 c_2 + a_3 c_3 + a_4 c_4$ with 
$\vb{a} \in \mathbb{R}^4$ and $\|\vb{a}\| = 1$.
Calculating the matrix explicitly, and using that
\[
c_1 = J(X, X), \quad c_2 = J(Y, Y), \quad c_3 = J(Z,Z), \quad c_4 = J(-i \one,i \one),
\]
one finds that it equals $J(A,B)$
with
\[A =\begin{pmatrix} a_3 -i a_4 & a_1 -i a_2 \\ a_1 +i a_2 & - (a_3 +i a_4)  \end{pmatrix}
\qquad \text{ and } \qquad
B 
=\begin{pmatrix} a_3 +i a_4 & a_1 -i a_2 \\ 
a_1 +i a_2 & - a_3 +i a_4 \end{pmatrix} = A^\dagger.
\]
Moreover, the determinant of $A$ equals:
\begin{align*}
    |A| &= - (a_3 - i a_4)(a_3 + i a_4) - (a_1 -i a_2)(a_1 +i a_2)
    \\ &= - (a_3^2 - (i a_4)^2) - (a_1^2 - (ia_2)^2)  \\
    &= - (a_1^2 + a_2^2 + a_3^2 + a_4^2) 
    \\&= -1 .
\end{align*}

Conversely, assume we have an odd matrix of the form $J(A, A^\dagger)$ with $|A| = -1$.
From the determinant condition, the matrix $A$ has the form
\[A = \begin{pmatrix} z & w \\ \Bar{w} & -\Bar{z} \end{pmatrix}
=
\begin{pmatrix} \Re(z) + i \Im(z) & \Re(w) + i \Im(w) \\ \Re(w) - i \Im(w) & -\Re(z) + i \Im(z) \end{pmatrix}.
\]
for some $w, z \in \CC$ with $|w|^2+|z|^2=1$.
Setting
\[a_1 = \Re(w), \quad a_2 = - \Im(w), \quad a_3 = \Re(z), \quad a_4 = -\Im(z), \]
we see that $J(A, A^\dagger) = a_1c_1 + a_2 c_2 + a_3 c_3 + a_4c_4$ and $\|\vb{a}\| = |w|^2+|z|^2 = 1$. 
\end{proof}

We now obtain a characterisation of all the two-qubit gates at higher levels of the hierarchy.

\begin{proposition} \label{prop:2qubit_hierarchy_characterisation}
For any $k \geq 2$, $\GMq{2}{k}$ is the set of unitaries of the form $G(A,B)$ or $J(A,B)$ with $|A|^{2^{k-2}} = |B|^{2^{k-2}}$.
\end{proposition}
\begin{proof}
Let $U=G(A,B)$ be an even two-qubit gate where $A, B \in \U(2)$.
We consider its action by conjugation on an arbitrary Majorana operator $c_\mu$ -- and, in fact, on any two-qubit gate in the first level of the hierarchy.  By \Cref{prop:2qubit_firstlevel_characterisation}, such an operator has the form $J(C,C^\dagger)$ for 
$C$ be a one-qubit gate with $|C|=-1$.
Conjugating $J(C,C^\dagger)$ by $U$ yields
\begin{equation}\label{eq:conjugate_firstlevel_by_even}
U J(C,C^\dagger) U^\dagger \;=\; G(A,B)J(C,C^\dagger)G(A^\dagger,B^\dagger) \;=\; J(AC, BC^\dagger)G(A^\dagger, B^\dagger) \;=\; J(ACB^\dagger,BC^\dagger A^\dagger).
\end{equation}
Comparing the determinants of its components, 
\[
\frac{|ACB^\dagger|}{|B{C^\dagger}{A^\dagger}|}
\;=\;
\frac{|A| \, |B^\dagger|}{|B| \, |A^\dagger|}
\;=\;
\frac{|A| \, \overline{|B|}}{|B| \, \overline{|A|}}
\;=\;
\frac{|A|^2}{|B|^2}.
\]
where the last equality follows from unitarity of $A$ and $B$, which implies that their determinants have modulus one and thus are complex numbers $z$ satisfying $\bar{z} = 1/z$. 
For an odd gate $U=J(A,B)$ one can perform a similar calculation.

The proof of the statement now proceeds by induction.
The base case, $k=2$, follows from the definition of matchgates.
For the induction step, suppose that the result holds for a given $k$.
By the induction hypothesis, the gate in \Cref{eq:conjugate_firstlevel_by_even} is in $\GMq{2}{k}$ if and only if $|A C B^\dagger|^{2^{k-2}} \!= |B C^\dagger A^\dagger|^{2^{k-2}}$,
which by the calculation above is equivalent to 
$(|A|^2)^{2^{k-2}} \!= (|B|^2)^{2^{k-2}}$ and so to
$|A|^{2^{k-1}} \!= |B|^{2^{k-1}}$.
\end{proof}

Characteristic examples that illustrate this result are the $\SWAP = G(\one,X)$  and $CZ = G(Z, \one)$ gates, which are third-level gates and thus satisfy $|A|  = - |B|$.

\begin{remark}[On the efficiency of the two-qubit protocol] \label{rem: efficiency_2_qubit_protocol}
The characterisation provided by \Cref{prop:2qubit_hierarchy_characterisation} has a neat consequence for the structure of each level of the two-qubit matchgate hierarchy, namely that each of the sets $\GMq{2}{k}$ is a subgroup of $U(2^n)$.
This, in turn, has a significant impact on the efficiency of \Cref{sec:teleportation_protocol}'s protocol when implementing a two-qubit gate $U$ at level $k$ of the hierarchy.
Recall that the necessary post-measurement correction is the conjugation of \Cref{correction_2} by $U$.
In general, this is a product of up to four gates that belong to a strictly lower level of the matchgate hierarchy than $k$.
Each of these (up to four) gates may, in turn, if they are not yet matchgates, require further application of the gate teleportation protocol consuming further resource magic states.
This potentially four-way branching spells trouble, as the overall efficiency for implementing a gate in the hierarchy appears to scale exponentially with the level $k$. The fact that $\GMq{2}{k}$ is a group means that the multiplication of the (up to four) correction gates itself belongs to the hierarchy at a level strictly lower than $k$.
Implementing it may still require further application of the protocol, but the fact that the correction is now a single gate eliminates the branching that could occur when descending the levels of the hierarchy.
Consequently, in the worst case, the protocol scales linearly with the level $k$ of the hierarchy when implementing two-qubit gates.
\end{remark}

\subsection{Matchgate-equivalence classes}

In a setting where matchgates are treated as `free' operations,
it is natural to study unitary transformations up to an equivalence relation that identifies those that are interconvertible resources.
Two gates are said to be \emph{(generalised-)matchgate-equivalent} which can be obtained from one another by multiplying on both sides with (generalised) matchgate circuits.

\Cref{prop:2qubit_hierarchy_characterisation} provides the basis for a classification of two-qubit (two-mode) fermionic gates in the hierarchy under such equivalence relations.
In summary, any even two-qubit gate $G(A,B)$ is classified up to matchgate equivalence by the phase $\phi \in [0,2\pi)$ given by $e^{i\phi} = |A|/|B|$,
with the $k$-level gates corresponding to the $2^{k-2}$-th roots of unity.
This correspondence is depicted in \Cref{fig:levels_unit_circles}.

\begin{proposition}\label{prop:2qubit_even_MGequiv}
Any two-qubit even unitary gate is matchgate-equivalent to the controlled-phase gate $C_\phi \defeq G(P_\phi, \one)$,
where $P_\phi = \diag(1,e^{i\phi})$, for a unique phase $\phi \in [0,2\pi)$.
\end{proposition}
\begin{proof}
Given an even unitary gate $G(A,B)$, both $A$ and $B$ are unitaries and thus their determinants have unit modulus.
Putting $e^{i \phi} = |A|/|B|$,
the gate
$G(P_\phi A^\dagger,B^\dagger)$
is a matchgate, since
\[
|P_\phi A^\dagger|
\;=\;
|P_\phi| \, |A^\dagger|
\;=\;
e^{i \phi} \overline{|A|}
\;=\;
\frac{|A|}{|B|}\,\frac{1}{|A|}
\;=\;
\frac{1}{|B|}
\;=\;
\overline{|B|}
\;=\;
|B^\dagger| .
\]
Moreover,
\[G(P_\phi A^\dagger,B^\dagger) \, G(A,B) \;=\; G(P_\phi A^\dagger A,B^\dagger B) \;=\; G(P_\phi,\one)\]
by unitarity of $A$ and $B$.
Hence,
$G(A,B)$ is matchgate-equivalent to $G(P_\phi,\one)$ for $\phi = \arg(|A|/|B|)$.

For the uniqueness,
suppose that $C_\phi$ and $C_{\phi'}$ are matchgate equivalent, \ie
that there exist matchgates
 $G(C, D)$ and $G(E, F)$ such that $G(C, D) C_{\phi} G(E, F) = C_{\phi^{\prime}}$.
 Taking the determinants on both sides yields
\[
    |C| \, |E| \, e^{i \phi} = e^{i \phi^{\prime}}
    \qquad \text{ and } \qquad 
    |D| \, |F| = 1 \\
\]
Since $|C| = |D|$ and $|E| = |F|$, this leads to $e^{i \phi} = e^{i \phi^{\prime}}$, hence $\phi = \phi^{\prime} \pmod{2 \pi}$. 
In other words, multiplying a gate $G(A,B)$ by matchgates $G(C,D)$ and $G(E,F)$ on each side leaves the quotient of determinants of its components invariant, i.e. $G(C,D)G(A,B)G(E,F) = G(CAE,DBF)$ satisfies
\[\frac{|CAE|}{|DBF|} = \frac{|C| |A| |E|}{|D| |B| |F|} = \frac{|A|}{|B|} .\]
\end{proof}

\begin{figure}[htb]
\centering

\def\scalecirc{.893}
\begin{tikzpicture}
[scale=\scalecirc,cap=round,>=latex]
    \draw[->] (-1.4cm,0cm) -- (1.4cm,0cm)  node[right,fill=white] {{\footnotesize$\mathrm{Re}$}};
    \draw[->] (0cm,-1.4cm) -- (0cm,1.4cm)  node[above,fill=white] {{\footnotesize$\mathrm{Im}$}}; 
    \draw (0cm,0cm) circle(1.0cm);
    \foreach \x in {0} { 
      \filldraw[black] (\x:1) circle (2pt); 
    }
    \path (0,-1.7) node {{\footnotesize $\GMq{2}{2}$}};
\end{tikzpicture}
\begin{tikzpicture}
[scale=\scalecirc,cap=round,>=latex]
    \draw[->] (-1.4cm,0cm) -- (1.4cm,0cm)  node[right,fill=white] {{\footnotesize$\mathrm{Re}$}};
    \draw[->] (0cm,-1.4cm) -- (0cm,1.4cm)  node[above,fill=white] {{\footnotesize$\mathrm{Im}$}}; 
    \draw (0cm,0cm) circle(1.0cm);
    \foreach \x in {0,180} { 
      \filldraw[black] (\x:1) circle (2pt); 
    }
    \foreach \x in {0} { 
      \filldraw[gray!75!white] (\x:1) circle (2pt); 
    }
    \path (0,-1.7) node {{\footnotesize $\GMq{2}{3}$}};
\end{tikzpicture}
\begin{tikzpicture}
[scale=\scalecirc,cap=round,>=latex]
    \draw[->] (-1.4cm,0cm) -- (1.4cm,0cm)  node[right,fill=white] {{\footnotesize$\mathrm{Re}$}};
    \draw[->] (0cm,-1.4cm) -- (0cm,1.4cm)  node[above,fill=white] {{\footnotesize$\mathrm{Im}$}}; 
    \draw (0cm,0cm) circle(1.0cm);
    \foreach \x in {0, 90, 180, 270} { 
      \filldraw[black] (\x:1) circle (2pt); 
    }
    \foreach \x in {0,180} {
      \filldraw[gray!75!white] (\x:1) circle (2pt); 
    }
    \path (0,-1.7) node {{\footnotesize $\GMq{2}{4}$}};
\end{tikzpicture}
\begin{tikzpicture}
[scale=\scalecirc,cap=round,>=latex]
    \draw[->] (-1.4cm,0cm) -- (1.4cm,0cm)  node[right,fill=white] {{\footnotesize$\mathrm{Re}$}};
    \draw[->] (0cm,-1.4cm) -- (0cm,1.4cm)  node[above,fill=white] {{\footnotesize$\mathrm{Im}$}}; 
    \draw (0cm,0cm) circle(1.0cm);
    \foreach \x in {0, 45, 90, 135, 180, 225, 270, 315} { 
      \filldraw[black] (\x:1) circle (2pt); 
    }
    \foreach \x in {0, 90, 180, 270} {
      \filldraw[gray!75!white] (\x:1) circle (2pt); 
    }
    \path (0,-1.7) node {{\footnotesize $\GMq{2}{5}$}};
\end{tikzpicture}
\begin{tikzpicture}
[scale=\scalecirc,cap=round,>=latex]
    \draw[->] (-1.4cm,0cm) -- (1.4cm,0cm)  node[right,fill=white] {{\footnotesize$\mathrm{Re}$}};
    \draw[->] (0cm,-1.4cm) -- (0cm,1.4cm)  node[above,fill=white] {{\footnotesize$\mathrm{Im}$}}; 
    \draw (0cm,0cm) circle(1.0cm);
    \foreach \x in {0, 22.5, 45, 67.5, 90, 112.5, 135, 157.5, 180, 202.5, 225, 247.5, 270, 292.5, 315, 337.5} { 
      \filldraw[black] (\x:1) circle (2pt); 
    }
    \foreach \x in {0, 45, 90, 135, 180, 225, 270, 315} {
      \filldraw[gray!75!white] (\x:1) circle (2pt); 
    }
    \path (0,-1.7) node {{\footnotesize $\GMq{2}{6}$}};
\end{tikzpicture}
\caption{Illustration of the correspondence between equivalence classes of two-qubit $k$-level gates and $2^{k-2}$-th roots of unity.
The complex number $e^{i \phi}$ represents the equivalence class of the gate $C_\phi$ consisting of the gates $G(A,B)$ with $|A|/|B|=e^{i \phi}$.
Each image depicts the equivalence classes that belong to a level of the matchgate hierarchy
(from the second, $\GMq{2}{2}$, to the sixth, $\GMq{2}{6}$), with the darker dots indicating those that appear for the first time at the level at hand.}
\label{fig:levels_unit_circles}
\end{figure}
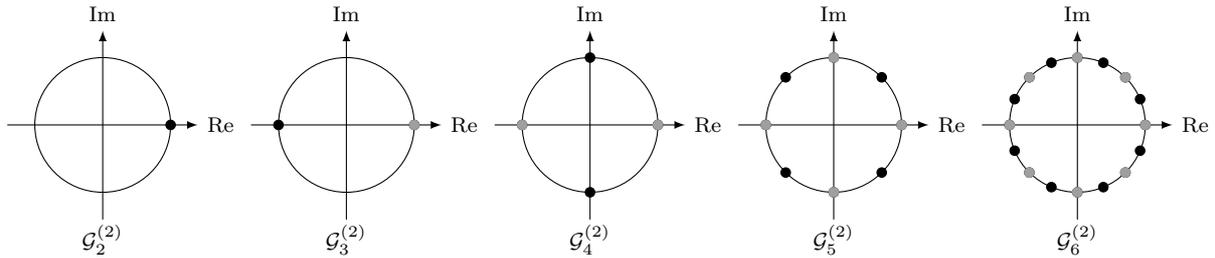

\begin{corollary}\label{cor:2qubit_evenhierarchy_MGequi}
      There are $2^{k-2}$ classes of even two-qubit $k$-level gates, \ie gates in $\Mq{2}{k}$, under matchgate equivalence.
      Representatives for each class are given by the gates $C_\phi$ for $\phi \in \setdef{\frac{2 j \pi}{2^{k-2}}}{j = 0, \ldots, 2^{k-2}-1}$ so that $e^{i\phi}$ are the
      $2^{k-2}$-th roots of unity.
  \end{corollary}

For odd gates, the picture is slightly more complicated.
We know that an arbitrary two-qubit odd gate $J(A,B)$ can be decomposed as $G(A,B) \, J(\one,\one)$.
This provides an already mentioned bijection between even and odd gates,
and thus there are corresponding equivalence classes of odd $k$-level gates under matchgate equivalence.
But considering generalised-matchgate equivalence, where odd matchgate circuits are also treated as ``free'' operations,
there are some further identifications.
Note that $J(A,B) = G(A,B) \, J(\one,\one) = J(\one,\one) \, G(B,A)$.
Therefore, in terms of the quotient $|A|/|B|$ of a gate $G(A,B)$,
multiplying by an odd matchgate might either leave it invariant or invert it, changing it to $|B|/|A|$.
Consequently, the equivalence classes of $C_\phi$ and $C_{-\phi}$ collapse.
This discussion is summarised by the following statements.

\begin{proposition}
Any two-qubit fermionic unitary gate is matchgate-equivalent to the controlled-phase gate $C_\phi \defeq G(P_\phi, \one)$,
where $P_\phi = \diag(1,e^{i\phi})$, for a unique phase $\phi \in [0,\pi]$.
\end{proposition}

\begin{corollary}\label{cor:2qubit_hierarchy_GMGequi}
      There are $2^{k-3}+1$ classes of fermionic two-qubit $k$-level gates, \ie gates in $\GMq{2}{k}$, under generalised matchgate equivalence.
      Representatives for each class are given by the gates $C_\phi$ for
      $\phi \in \setdef{\frac{2 j \pi}{2^{k-2}}}{j = 0, \ldots, 2^{k-3}}$
      so that $e^{i\phi}$ are the
      $2^{k-2}$-th roots of unity with nonnegative imaginary part.
  \end{corollary}

\section{Hierarchical fermionic Stone--von Neumann theorem} \label{sec:fermionic_stone_von_neumman_theorem}

We prove a hierarchy-aware fermionic version of the Stone--von Neumann theorem;
cf.\ the analogous result for the Clifford hierarchy in Ref.~\cite{de_Silva_2021}.
This can be seen as a stepping stone to elucidating the structure of the matchgate hierarchy for arbitrary number of qubits.
We start by considering a general statement without reference to the hierarchy, from which the hierarchical version follows.

Consider a set $\{c_\mu\}_{\mu=1}^{2n}$ of $2n$ Hermitian operators satisfying the CAR. The monomials of the form
$c_{\mu_1} \cdots c_{\mu_\lmu}$ with $\mu_1 < \cdots < \mu_\lmu$ are linearly independent
and form a basis of the space $M_{2^n}(\CC)$  of $n$-qubit matrices \cite[Theorem 2]{lee1947clifford}.

\begin{theorem}[fermionic Stone--von Neumann theorem]\label{thm:fermionicSvN}
Given two sets $\{c_{\mu}\}_{\mu=1}^{2n}$ and $\{d_{\mu}\}_{\mu=1}^{2n}$ of
$n$-qubit operators which satisfy the canonical anticommutation relations (CAR), there is an $n$-qubit unitary $U \in \U(2^n)$, unique up to a phase, such that $U^\dagger c_{\mu}U = d_{\mu}$.
\end{theorem}
\begin{proof} Set $\phi(c_{\mu_1} \cdots c_{\mu_\lmu}) = d_{\mu_1}\cdots d_{\mu_\lmu}$ for $1 \leq \mu_1 < \cdots < \mu_\lmu \leq 2n$.
By the observation above, this sends a basis of $M_{2^n}(\CC)$ to another, so it extends uniquely to a linear automorphism $\phi \colon M_{2^n}(\CC) \to M_{2^n}(\CC)$.

We now show that it preserves adjoints and multiplication. It is enough to check this on basis elements, \ie on the monomials of the form above.
For the adjoint, note that given $\mu_1 < \cdots < \mu_\lmu$,
\[
     \left(c_{\mu_1}  \dots c_{\mu_\lmu}\right)^{\dagger}
    \;\;=\;\; 
     c_{\mu_\lmu}^\dagger  \cdots c_{\mu_1}^\dagger
    \;\;=\;\;
     c_{\mu_\lmu}  \cdots c_{\mu_1}
    \;\;=\;\;
     (-1)^{\binom{\lmu}{2}}\, c_{\mu_1}  \cdots c_{\mu_\lmu} .
\]
The derivation only uses the fact that $\{c_\mu\}$ satisfy the CAR, and so the same holds for $\{d_\mu\}$.
Hence,
\begin{align*}
\phi\left(\left(c_{\mu_1}  \dots c_{\mu_\lmu}\right)^{\dagger}\right)
 \;\;=\;\; &
\phi \left( (-1)^{\binom{\lmu}{2}}\, c_{\mu_1}  \cdots c_{\mu_\lmu}\right)
\\ \;\;=\;\; &
(-1)^{\binom{\lmu}{2}}\, \phi\left(c_{\mu_1}  \cdots c_{\mu_\lmu}\right)
\\  \;\;=\;\; &
(-1)^{\binom{\lmu}{2}}\, d_{\mu_1}  \cdots d_{\mu_\lmu}
\\  \;\;=\;\; &
\left( d_{\mu_1}  \cdots d_{\mu_\lmu}\right)^\dagger
\\  \;\;=\;\; &
\phi\left(c_{\mu_1}  \cdots c_{\mu_\lmu}\right)^\dagger
\end{align*}
For the multiplication, similarly note that given $\mu_1 < \cdots < \mu_\lmu$ and $\nu_1 < \cdots < \nu_\lnu$,
\[
(c_{\mu_1} \dots c_{\mu_\lmu})\,(c_{\nu_1} \dots c_{\nu_\lnu})
=
(-1)^{\#(\mu > \nu)}c_{\lambda_1} \cdots c_{\lambda_\llam}
\]
where  the set $\{\lambda_1, \dots, \lambda_\llam\}$, with $\lambda_1 < \dots < \lambda_\llam$, is the symmetric difference of the sets $\{\mu_1, \dots, \mu_\lmu\}$ and $\{\nu_1, \dots, \nu_\lnu\}$
and
$\#(\mu > \nu) = |\setdef{(i,j)}{\mu_i > \nu_j}|$
counts the number of swaps needed to reorder the monomial.
Again,  this derivation uses the CAR alone, so the same holds for $\{d_\mu\}$, too.
Hence, 
\begin{align*}
    \phi\left((c_{\mu_1} \dots c_{\mu_\lmu})\,(c_{\nu_1} \dots c_{\nu_\lnu})\right)
    \;\;=\;\; & \phi\left((-1)^{\#(\mu > \nu)} \, c_{\lambda_1} \cdots c_{\lambda_\llam}\right)
    \\ \;\;=\;\; & (-1)^{\#(\mu > \nu)} \, \phi\left(c_{\lambda_1} \cdots c_{\lambda_\llam}\right) 
    \\ \;\;=\;\; & (-1)^{\#(\mu > \nu)}d_{\lambda_1} \cdots d_{\lambda_\llam}
    \\ \;\;=\;\; &
     (d_{\mu_1} \dots d_{\mu_\lmu})\,(d_{\nu_1} \cdots d_{\nu_\lnu}) 
    \\ \;\;=\;\; &
    \phi(c_{\mu_1} \cdots c_{\mu_\lmu})\,\phi(c_{\nu_1} \cdots c_{\nu_\lnu}) .
\end{align*}

We have established that $\phi$ is a $*$-automorphism of the matrix algebra $M_{2^n}(\CC)$.
Therefore, by the Skolem--Noether theorem, it is an inner automorphism, induced by a unitary $U$.
\end{proof}

The unitary $U$ in the above statement is given explicitly as follows:
for each $\vz \in \ZZ_2^{2n}$, the action of $U$ on the corresponding computational basis vector satisfies
\[
    U \ket{\vz} = \tc_1^{z_1} \tc_3^{z_2} \cdots \tc_{2n-1}^{z_n} \; U \ket{\vec{0}} = \left(\textstyle\prod\nolimits_{k=1}^{n} \tc_{2k-1}^{z_k}\right) \;  U\ket{\vec{0}} ,
\]
and
$U \ket{\vec{0}}$ is the simultaneous $+1$-eigenvector 
of the operators $(-i\tc_{2k-1}\tc_{2k})$ for all $k \in \{1, \ldots, n\}$.

\begin{corollary}\label{cor:fSvN-bijective}
There is a bijective correspondence between $n$-qubit unitary maps up to phase and tuples of $n$-qubit operators $(\tilde{c}_{\mu})_{\mu=1}^{2n}$ satisfying the CAR.
\end{corollary}
\begin{proof}
    Given a unitary $U$, take $\tilde{c}_\mu \defeq U^{\dagger} c_\mu U$.
    These operators satisfy the CAR, as such relations are preserved under conjugation, an (inner) automorphism.
    Moreover, its conjugation action on the Majorana operators, since these generate $M_{2^n}(\CC)$ as an algebra,
    fully determines the action on the whole matrix algebra. Hence, it determines the unitary $U$ up to phase,
    showing that the map from unitaries up to phase to CAR tuples is injective.
    \Cref{thm:fermionicSvN} above shows that it is surjective.
\end{proof}

One can restrict these results to characterise the fermionic (even or odd) unitaries.
By \Cref{lemma:unitary_parity_preserving_conjugation_implies_evenorodd}, these are precisely the unitaries that preserve matrix parity under conjugation.
We thus arrive at the following versions of
\Cref{thm:fermionicSvN} and \Cref{cor:fSvN-bijective}

\begin{corollary} \label{cor:fSvN_odd_tuples}
Given two sets $\{c_{\mu}\}_{\mu=1}^{2n}$ and $\{d_{\mu}\}_{\mu=1}^{2n}$ of odd
$n$-qubit operators that satisfy the CAR,
there is an $n$-qubit fermionic unitary $U \in \U(2^n)$, unique up to a phase, such that $U^\dagger c_{\mu}U = d_{\mu}$.
\end{corollary}

\begin{corollary} \label{cor:fSvN-bijective_odd_tuples}
There is a bijective correspondence between $n$-qubit fermionic unitary maps up to phase and tuples of $n$-qubit odd operators $(\tilde{c}_{\mu})_{\mu=1}^{2n}$ satisfying the CAR.
\end{corollary}

The following version then follows immediately from the definition of the matchgate hierarchy.

\begin{corollary}[Hierarchical fermionic Stone--von Neumann theorem]
Given a set $\{\tc_{\mu}\}_{\mu=1}^{2n}$ of odd operators in $\GMn{k}$ that satisfy the canonical anticommutation relations,
there exists a unitary $U \in \GMn{k+1}$, unique up to a phase,
such that $U^\dagger c_{\mu} U = \tc_{\mu}$ for all $\mu$. 
This $U$ is given explicitly as follows:
for each $\vec{z} \in \ZZ_2^n$, the action of $U$ on the corresponding computational basis vector satisfies
\[
    U \ket{\vec{z}} = \tc_1^{z_1} \tc_3^{z_2} \cdots \tc_{2n-1}^{z_n} \; U \ket{\vec{0}} = \left(\textstyle\prod\nolimits_{k=1}^{n} \tc_{2k-1}^{z_k}\right) \;  U\ket{\vec{0}} ,
\]
and
$U \ket{\vec{0}}$ is the simultaneous $+1$-eigenvector 
of the operators $(-i\tc_{2k-1}\tc_{2k})$ for all $k \in [n]$.
\end{corollary}

\begin{corollary}
There is a bijection between $\GMn{k+1}$ and $2n$-tuples in $\GMn{k} \cap \Odd{n}$ that satisfy the CAR.
Consequently, $\GMn{k}$ spans a subspace of $2^n \times 2^n$ matrices of dimension at most $(2n)^k$.
\end{corollary}

\section{Outlook}\label{sec:outlook}

In this work we adapted the concept of the Clifford hierarchy to the context of matchgate circuits:
we presented a gate teleportation protocol and an associated hierarchy of fermionic unitary gates, the matchgate hierarchy.
Any $n$-qubit gate in the hierarchy can be deterministically implemented using adaptive matchgate circuits and matchgate-magic states.
Moreover, we gave a complete characterisation of the gates in the matchgate hierarchy for two qubits, whereby matchgate-equivalence classes of even $k$-level gates correspond to the $2^{k-2}$-th roots of unity.
As a consequence of our particular characterisation, we concluded that the number of matchgate-magic states needed as a resource for our protocol grows linearly with the level of the target gate in the hierarchy.
For an arbitrary number of qubits we proved a `hierarchy-aware' fermionic Stone--von Neumann theorem which paves the way for a better understanding of the structure of the matchgate hierarchy.

Various directions for follow-up research arise from the ideas developed in this work.
These fall broadly into two categories.
The first concerns further studying the structure of the matchgate hierarchy
and the gate teleportation protocol that underpins it, as well as variants.
The second relates to the broader context wherein magic state protocols are useful, that of fault-tolerant quantum computing, as well as to other potential applications of the matchgate hierarchy.

\begin{itemize}
\item We provided a complete characterisation of the matchgate hierarchy for two-qubit gates,
built upon a parametrisation of two-qubit fermionic (even and odd) unitaries. A natural question is
to search for a similar characterisation for higher numbers of qubits.

One possible route is by relaxing matchgate identities.
These are equations on the entries of an $n$-qubit unitary matrix that provide necessary and sufficient conditions for it to be Gaussian, \ie implementable by matchgate circuit.
It is conceivable that a relaxation of such identities can be used to identify gates in the matchgate hierarchy.
Note that our characterisation for two-qubit gates in \Cref{sec:2qubit_characterisation} can be thought of in this fashion:
the matchgate identities for two-qubit gates essentially check the matching determinants condition, which is relaxed in a controlled fashion for gates at higher levels of the hierarchy.

\item A possible pay-off from such a characterisation, just as in the two-qubit case, would be a finer analysis of the resource complexity of the gate teleportation protocol for $n$-qubit gates as one goes up the levels of the hierarchy, itself a crucial open question.

\item Another line of inquiry concerns the structure of the matchgate hierarchy.
One may draw inspiration from the extensive work on the Clifford hierarchy.
For instance, Pllaha et al.~\cite{Pllaha2020unweylingclifford} examined the structure gates in the Clifford hierarchy through their Pauli (or Weyl)
expansion into linear combinations of Pauli operators,
in particular characterising the Pauli support of second- and third-level gates.
A similar approach could be followed for the matchgate hierarchy, using the Majorana expansion of unitaries
into a linear combination of monomials of Majoranas, \ie elements of the Majorana group (see paragraph following \Cref{eq:CARgroup}).
One could build on recent findings regarding the structure of this group \cite{bettaque2024}.
A partial result we are able to report in this direction is that the conjugation of a Majorana by a two-qubit third-level gate is supported on degree-3 Majorana monomials; it is not yet clear how this generalises to higher number of qubits.

Other analogies (or disanalogies) with the Clifford hierarchy could be explored.
More broadly, mapping out similarities and differences between the Clifford and matchgate hierarchies
could be fruitful for drawing out some common features,
aspiring towards a general framework that encompasses 
both hierarchies and potentially others in the same spirit, offering a unified perspective on gate hierarchies.

\item The state $B^{(n)}\!\ket{\mathbf{0}}$ could be thought of as the would-be matchgate-magic state for the $n$-qubit identity operator $\one$, in that running our gate teleportation protocol using it as the resource state would implement the identity, as calculated for the case $n=2$ in \Cref{eq:protocol_snake}. Of course, this is not a proper matchgate-magic state since the identity operator is itself Gaussian.
Still, \Cref{eq:protocol_snake} is the core of the gate teleportation protocol for an arbitrary $U$.
It is reminiscent of the standard qubit teleportation protocol, neatly -- and strikingly visually -- captured by the `yanking' equations in categorical quantum mechanics \cite{abramsky2004categorical}. A question is whether $B^{(n)}\!\ket{\mathbf{0}}$
provides the `cup' -- that is, a compact closed structure -- in a suitable monoidal category for fermionic maps; cf.\ \cite{felice2019diagrammatic,carette2023}

\item Parafermions \cite{fradkin1980parafermions,gungordu2014parafermion, hutter_2016} are a generalisation of Majorana fermions
whereby the canonical anticommutation relations of \Cref{eq:CARgroup} give way to
$c_\mu^d = \one$ and $c_\mu c_\nu = \omega c_\nu c_\mu$ for $\mu < \nu$,
where $\omega = e^{2 \pi i / d}$ is the primitive $d$th root of unity, for a fixed arbitrary $d \in \NN$
(Majorana fermions are the $d = 2$ case).
The Jordan--Wigner transform generalises to a mapping between parafermionic operators and qudit Paulis.
It would be interesting to explore to what extent the story carries over to that case.

\item In the stabiliser world, gates in the Clifford hierarchy that are diagonal up to Clifford equivalence (\ie of the form $C_1DC_2$ for $C_i$ Clifford and $D$ diagonal), named \emph{semi-Clifford gates},
admit a more efficient gate teleportation protocol \cite{Zhou_2000} (see also \cite{Appleby_2012,Zeng_2008}).
It uses magic states with half the number of qubits: implementing a semi-Clifford $n$-qubit unitary $U$ consumes an $n$-qubit state rather than a $2n$-qubit state: essentially, the state $U\!\ket{+}$ instead of the state $(\one \otimes U)\!\ket{\Phi_{\text{Bell}}}$.
An open question is whether a similarly more efficient protocol exists for a subclass of gates in the matchgate hierarchy.

\item
The Clifford hierarchy plays a central role in a leading model for fault-tolerant quantum computing
built upon stabiliser circuits (which are by themselves classically simulable) adjoined with so-called magic states,
pre-prepared nonstabiliser states that encode a specific non-Clifford operation (sufficient to reach quantum universality) to be implemented via gate teleportation.
Underpinning this model is a key feature of Clifford operations that they can be implemented fault-tolerantly:
there are Clifford-generating gate sets that can be implemented transversally on logical qubits encoded by a specific stabiliser error-correcting code.
There are studies on error correction to protect quantum information against fermionic errors, notably Majorana fermion codes \cite{Bravyi_2010}, which are also mathematically reminiscent of stabiliser codes but based on the Majorana instead the Clifford group.
However, a fault-tolerant model for (transversal) implementation of matchgate circuits is yet to be developed.

\item 
Another crucial ingredient in such a model is the preparation of the resource states to be used by the matchgate circuit.
On the one hand, it is pertinent to find circuits that prepare magic states fault-tolerantly, using operations that can be implemented transversally under an error-correcting code, typically a different code from the one to be used to implement matchgate circuits, in analogy to what can be done for the usual (Clifford) magic states.
On the other, it would be relevant to study distillation protocols for matchgate-magic states.

\item
Finally, Ref.~\cite{matchgate_hierarchy_sergii} uses the matchgate hierarchy in the context of learning unknown fermionic Gaussian unitaries, analogously to the results of Ref.~\cite{Low_2009} for the Clifford hierarchy.
Given the wealth of applications of the Clifford hierarchy, it is natural to seek other analogous applications for the matchgate hierarchy. 

\end{itemize}


\section*{Acknowledgements}

AB would like to thank Anita Camillini, Raman Choudhary, Antonio Ruiz-Molero, and Wilfred Salmon for valuable discussions. AB also thanks Simon Fraser University (SFU) for the kind hospitality and the group of NdS for many useful comments and fruitful discussions.  

The authors would also like to thank anonymous conference referees for helpful comments and for spotting an error in an earlier version of this work.

This work was supported by the European Union and NSERC -- Natural Sciences and Engineering Research Council of Canada through the Digital Horizon Europe project FoQaCiA, GA no. 101070558 (AB, RSB, NdS).
AB and RSB also acknowledge support from FCT -- Funda\c{c}\~ao para a Ci\^encia e a Tecnologia (Portugal) through PhD Grant SFRH/BD/151310 (AB) and through CEECINST/00062/2018 (RSB).
NdS also acknowledges support from the Canada Research Chair program, NSERC Discovery Grant RGPIN-2022-03103, and the Faculty of Science of Simon Fraser University.

\bibliographystyle{plain}
\bibliography{references}

\end{document}